%% file: main.tex
\documentclass{sig-alternate}
\pdfoutput=1
\usepackage{verbatim}
\usepackage{graphicx}
\usepackage{subfigure}
\usepackage{comment}
\usepackage{array}
\usepackage{mathtools}
\usepackage{amsmath}
\usepackage{color}
\usepackage{times}
\usepackage{multirow}
\usepackage{dsfont}
\usepackage{amsfonts}

\makeatletter
\newif\if@restonecol
\makeatother

\usepackage[linesnumbered, ruled, vlined]{algorithm2e}

\newcommand{\hide}[1]{} %hide
\newcommand{\vpara}[1]{\vspace{0.1in}\noindent\textbf{#1 }}
\newcommand{\para}[1]{\vspace{0.01in}\noindent\textbf{#1 }}
\newcommand{\secref}[1]{Section~\ref{#1}} %section reference
\newcommand{\beq}[1]{\vspace{-0.02in}\begin{equation}#1\end{equation}\vspace{-0.02in}}

\hide{
\usepackage[ruled]{algorithm2e}

\renewcommand{\algorithmcfname}{ALGORITHM}
\SetAlFnt{\small}
\SetAlCapFnt{\small}
\SetAlCapNameFnt{\small}
\SetAlCapHSkip{0pt}
\IncMargin{-\parindent}
}

\newdef{definition}{Definition}
\newdef{problem}{Problem}
\newdef{theorem}{Theorem}
\newdef{lemma}{Lemma}

\newcommand{\model}{Panther }

\begin{document}

\title{Panther: Fast Top-k Similarity Search in Large Networks}

\numberofauthors{1}
%
% Put no more than the first THREE authors in the \author command

% NOTE: All authors should be on the first page. For instructions
% for more than 3 authors, see:
% http://www.acm.org/sigs/pubs/proceed/sigfaq.htm#a18

\author{
%% e-mail address with \email.
\alignauthor Jing Zhang$^{\dag}$, Jie Tang$^{\dag\sharp}$, Cong Ma$^{\dag}$, Hanghang Tong$^{\ddag}$, Yu Jing$^{\dag}$, and Juanzi Li$^{\dag}$\\
       \affaddr{$^{\dag}$Department of Computer Science and Technology, Tsinghua University}\\
       $^{\sharp}$Tsinghua National Laboratory for Information Science and Technology (TNList)\\
       \affaddr{$^{\ddag}$School of Computing, Informatics, and Decision Systems Engineering, ASU}\\
       \email{\sf \{zhangjing12, ma-c11\}@mails.tsinghua.edu.cn, \{jietang, yujing5b5d,lijuanzi\}@tsinghua.edu.cn, hanghang.tong@asu.edu }
}

\maketitle
\sloppy

\input{abstract.tex}

% A category with only the three required fields
\category{H.3.3}{Information Search and Retrieval}{Text Mining}
\category{J.4}{Social Behavioral Sciences}{Miscellaneous}
\category{H.4.m}{Information Systems}{Miscellaneous}
\terms{Algorithms, Experimentation}

\keywords{Vertex similarity; Social network; Random path}

\input{intro.tex}

\input{problem.tex}

\input{approach.tex}
\input{exp.tex}

\input{related.tex}

\input{conclusion.tex}
\input{ack.tex}

%\clearpage

\bibliographystyle{abbrv}
%\small
\bibliography{references}  % sigproc.bib is the name of the Bibliography in this case

\end{document}

%% file: abstract.tex
\begin{abstract}
%Measuring similarity between verexes in a network has long been a focus of network mining. It is helpful in numerous applications including search engines and recommender systems.
Estimating similarity between vertices is a fundamental issue in network analysis across various domains, such as social networks and biological networks.
Methods based on common neighbors and structural contexts have received much attention. However, both categories of methods are difficult to scale up to handle large networks (with billions of nodes).
%Two types of similarity measures, neighborhood similarity and structure similarity, have been studied separately.
In this paper, we propose a sampling method that provably and accurately estimates the similarity between vertices. The algorithm is based on a novel idea of \textit{random path}, and an extended method is also presented, to enhance the structural similarity when two vertices are completely disconnected.
We provide theoretical proofs for the error-bound and confidence of the proposed algorithm.

We perform extensive empirical study and show that our algorithm can obtain top-$k$ similar vertices for any vertex in a network approximately 300$\times$ faster than state-of-the-art methods.
We also use identity resolution and structural hole spanner finding, two important applications in social networks, to evaluate the accuracy of the estimated similarities. Our experimental results demonstrate that the proposed algorithm achieves clearly better performance than several alternative methods.

\hide{
find an effective and efficient way to address the problem of finding top-$K$ similar vertices in terms of these two measures. Originally, we propose two methods $PathSim$ and $PathVec$ to correspondingly measure the neighborhood similarity and structure similarity. $PathSim$ is based on the simple assumption that ``two nodes are similar if they frequently appear on the same path'' and PathVec relies on the idea that ``vertices with similar topological characteristics will have similar distributions of tie strengths with other vertices''. Moreover, to deal with the intractability of this kind of calculation of similarity, we adopt a sampling based method to approximate the result. Theoretical guarantees have been made to justify the use of such an approximation algorithm. It further shows that the number of random paths needed to achieve a desired approximation level is independent of the network size. The experiments on various datasets demonstrate the effectiveness and efficiency of our methods.}
\end{abstract}

%% file: intro.tex
\section{Introduction}
\label{sec:intro}

Estimating vertex similarity is a fundamental issue in network analysis and also the cornerstone of many data mining algorithms such as clustering, graph matching, and object retrieval.
The problem is also referred to as \textit{structural equivalence} in previous work~\cite{Lorrain:71JMS},   and has been extensively studied in physics, mathematics, and computer science.
In general, there are two basic principles to quantify similarity between vertices.
The first principle is that  two vertices are considered structurally equivalent if they have many common neighbors in a network.
The second principle is that two vertices are considered structurally equivalent if they play the same structural role---this can be further quantified by  degree, closeness centrality,  betweenness, and other network centrality metrics~\cite{Freeman:79SN}.
Quite a few similarity metrics have been developed based on the first principle, e.g., the Jaccard index~\cite{Jaccard:1901} and Cosine similarity~\cite{baeza1999modern}. 
However, they estimate the similarity in a local fashion. Though some work such as SimRank~\cite{jeh:KDD06}, VertexSim~\cite{leicht:PRE2006}, and RoleSim~\cite{jin:KDD11}, use the entire network to compute similarity, they are essentially based on the transitivity of similarity in the network.
There are also a few studies that follow the second principle. For example,
Henderson et al.~\cite{henderson:KDD11} proposed a feature-based method, named ReFeX, to calculate vertex similarity by defining a vector of features for each vertex.

Despite much research on this topic, the problem remains largely unsolved. The first challenge is how to design a unified method to accommodate both principles. This is important, as in many applications, we do not know which principle to follow.
The other challenge is the efficiency issue. Most existing methods have a high computation cost.  SimRank results in a complexity of $O(I|V|^2\bar{d}^2)$, where $|V|$ is the number of vertices in a network; $\bar{d}$ is the average degree of all vertices; $I$ is the number of iterations to perform the SimRank algorithm. It is clearly infeasible to apply SimRank to large-scale networks. For example, in our experiments, when dealing with a network with  500,000 edges, even the fast (top-$k$) version of SimRank~\cite{lee:2012top} requires more than five days to complete the computation for all vertices (as shown in Figure~\ref{subfig:efficienty}).

Thus, our goal in this work is to design a similarity method that is flexible enough to incorporate different structural patterns (features) into the similarity estimation and to quickly estimate vertex similarity in very large networks.

\begin{figure*}
\centering
\mbox{
\subfigure[Top-$k$ similarity search]{\label{subfig:example}
\includegraphics[width=0.27\textwidth]{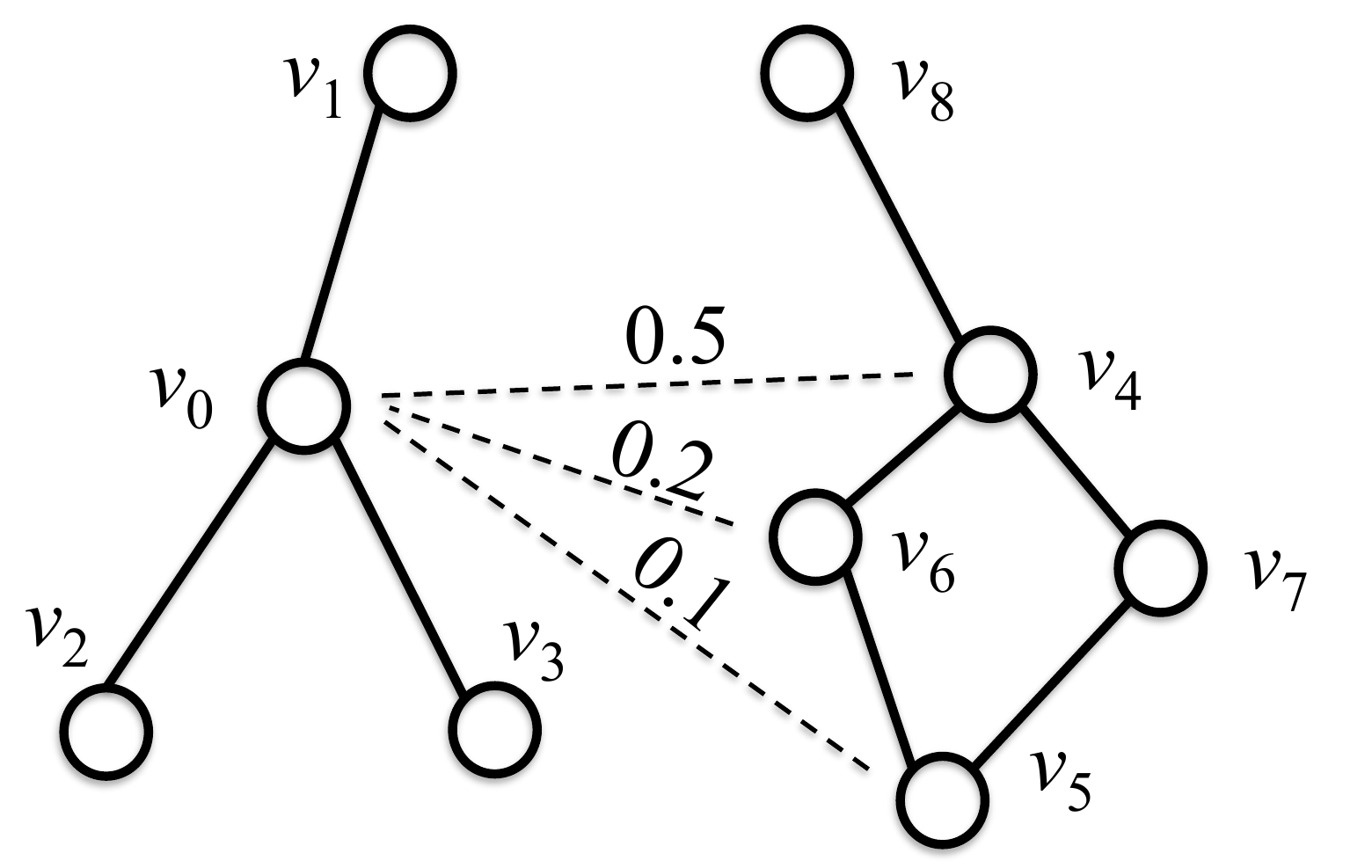}
}
\hspace{0.1in}
\subfigure[Efficiency Performance]{\label{subfig:efficienty}
\includegraphics[width=0.3\textwidth]{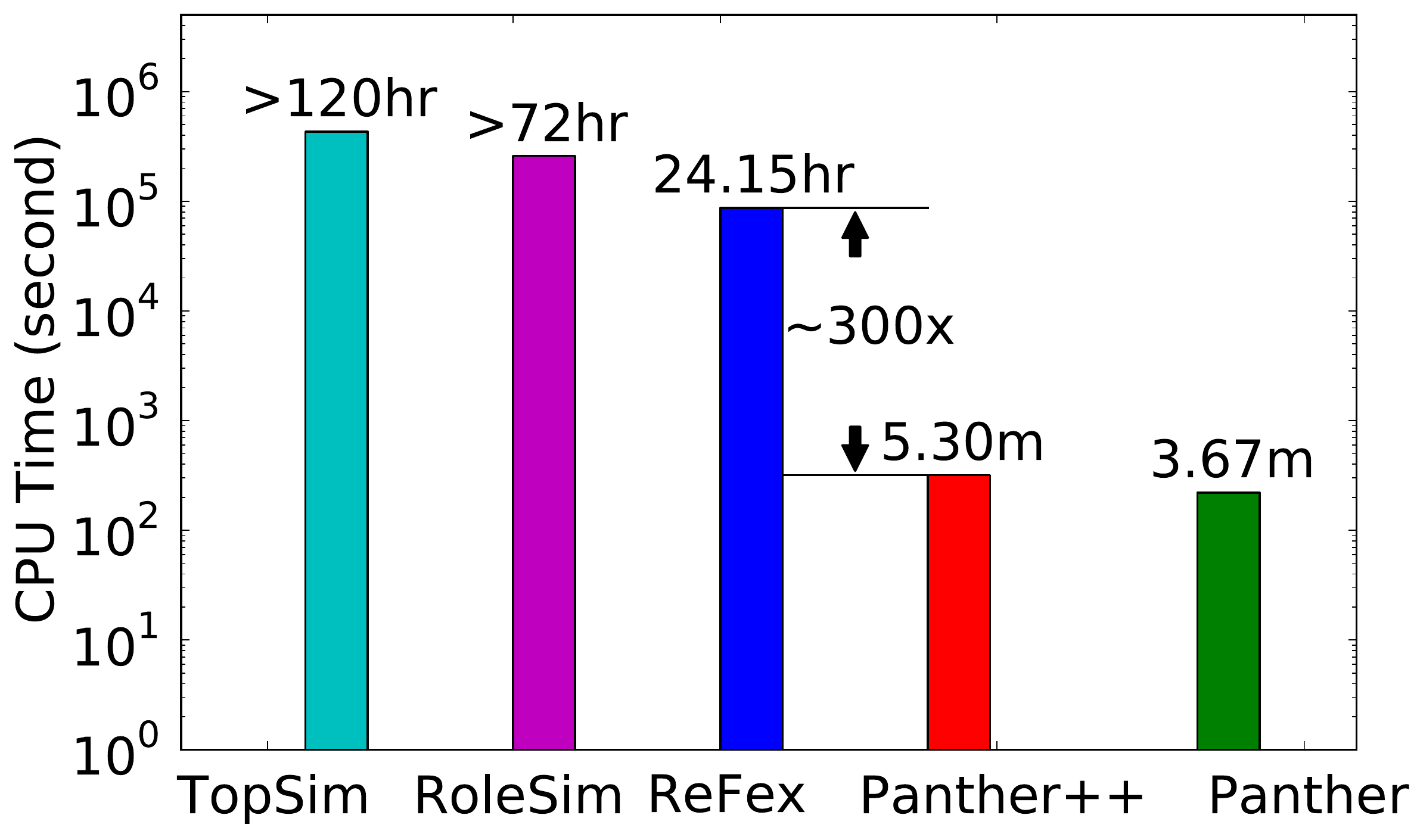}
}
\hspace{0.1in}
\subfigure[Application: Identity Resolution]{\label{subfig:identityresolution}
\includegraphics[width=0.3\textwidth]{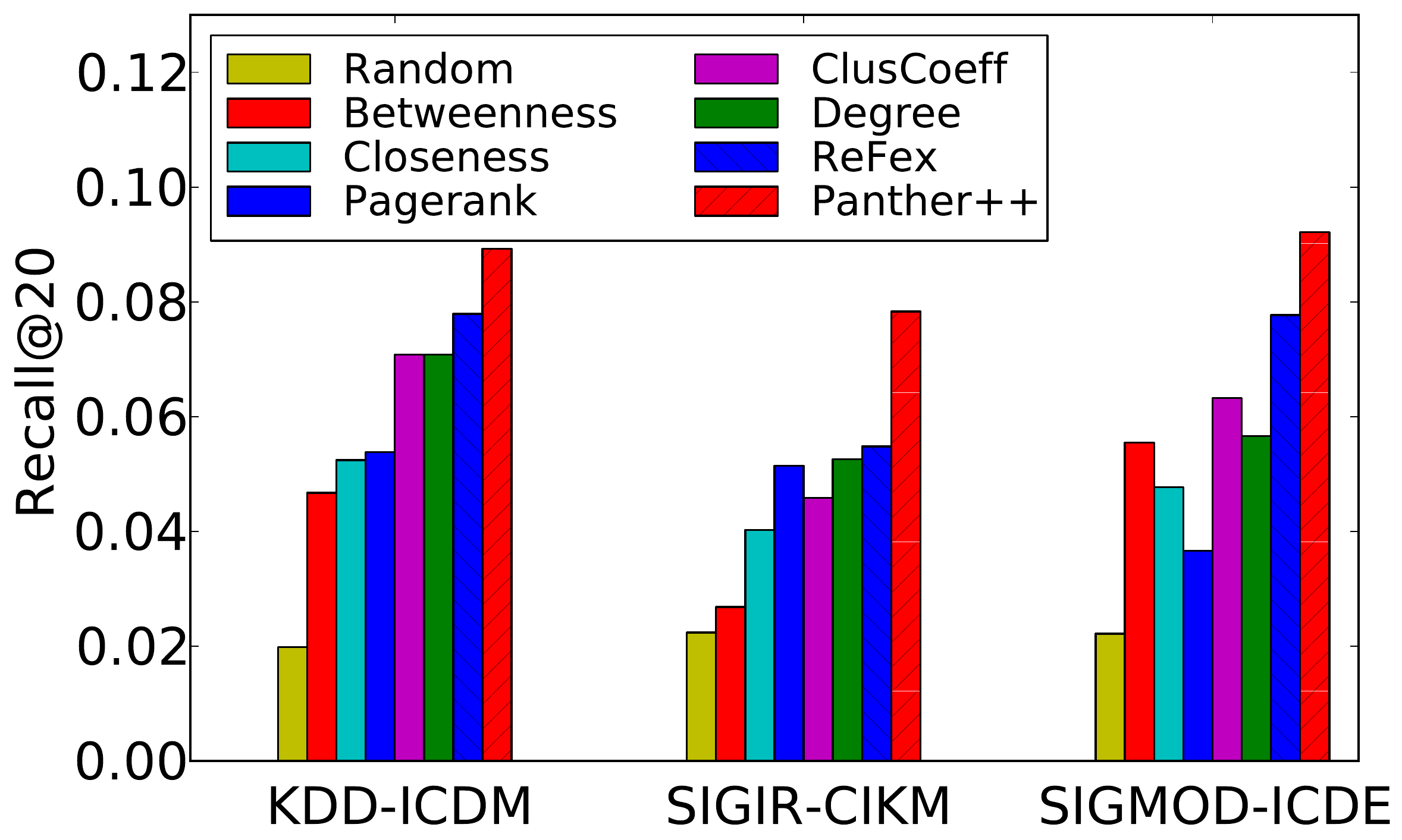}
}
}
%\vspace{-0.1in}
\caption{\label{fig:result} Example of top-$k$ similarity search across networks and performance comparison. (a) Top-$k$ similarity search across two disconnected networks; (b) Efficiency comparison of Panther and several comparison methods on a Tencent subnetwork of 443,070 vertices and 5,000,000 edges; and (c) Accuracy performance when applying Panther++ to identity resolution~\cite{gilpin:2013}, an important application in social network. Please refer to \S~\ref{sec:exp} for definitions of all the comparison methods in (b) and (c).}
\end{figure*}

%\vpara{Our Solution and Contributions.}
We propose a sampling-based method, referred to as Panther, that provably and quickly estimates the similarity between vertices.
The algorithm is based on a novel idea of \textit{random path}. Specifically, given a network, we perform $R$ random walks, each starting from a randomly picked vertex and walking $T$ steps. The idea behind this is that two vertices have a high similarity if they frequently appear on the same paths. We provide theoretical proofs for the error-bound and confidence of the proposed algorithm. Theoretically, we obtain that the sample size, $R = \frac{c}{\varepsilon^2} ( \log_2 \binom{T}{2} +1 +\ln \frac{1}{\delta})$, only depends on the path length $T$ of each random walk, for a given error-bound $\varepsilon$ and confidence level $1-\delta$.
To capture the information of structural patterns, we extend the proposed algorithm by augmenting each vertex with a vector of structure-based features.
The resultant algorithm is referred to as Panther++. Panther++ is not only able to estimate similarity between vertices in a connected network, but also capable of estimating similarity between vertices from disconnected networks. Figure~\ref{subfig:example} shows
an example of top-$k$ similarity search across two disconnected networks, where $v_4$, $v_6$ and $v_5$ are top-3 similar vertices to $v_0$.

We evaluate the efficiency of the methods on a microblogging network from Tencent\footnote{http://t.qq.com}.
Figure~\ref{subfig:efficienty} shows the efficiency comparison of Panther, Panther++, and several other methods.
Clearly, our methods are much faster than the comparison methods.

Panther++ achieves a 300$\times$ speed-up over the fastest comparison method on a Tencent subnetwork  of 443,070 vertices and 5,000,000 edges. Our methods are also scalable. 
Panther is able to return top-$k$ similar vertices for all  vertices in a network with 51,640,620 vertices and 1,000,000,000 edges. On average, it only need 0.0001 second to perform top-$k$ search for each vertex.

We also evaluate the estimation capability of Panther++. Specifically, we use identity resolution and top-$k$ structural hole spanner finding,  two important applications in social networks, to evaluate the accuracy of the estimated similarities. Figure~\ref{subfig:identityresolution} shows the accuracy performance of Panther++ and several alternative methods for identity resolution. Panther++ achieves clearly better performance than several alternative methods.
All codes and datasets used in this paper are publicly available\footnote{https://github.com/yujing5b5d/rdsextr}.

\vpara{Organization}  \secref{sec:problem} formulates the problem. In \secref{sec:approach}, we detail the proposed methods for top-$k$ similarity search, and provide theoretical analysis.
%\secref{sec:proof} gives theoretical justification for the use of sampling based approximation algorithm.
\secref{sec:exp} presents experimental results to validate the efficiency and effectiveness of our methods. \secref{sec:related} reviews the related work. Finally, \secref{sec:conclusion} concludes the paper.

\hide{
Given a vertex in a network, one important question of network mining is to find the vertices that are most similar to it. Measuring the similarity between vertices has numerous applications. For example, in search engines, vertex similarity can help rank webpages in the web network~\cite{BaezaYates:99}. In recommender systems, vertex similarity can be used to group objects, such as similar users or similar items~\cite{Goldberg:1992,Konstan:1997}. Another example is to utilize the similarity to facilitate de-anonymization of people in a social network.

A number of similarity measures have been studied, such as~\cite{henderson:KDD11, jeh:KDD06,jin:KDD11,leicht:PRE2006}. In general, vertex similarity can be divided into two types, which we name as \textit{neighborhood similarity} and \textit{structure similarity}. The first type is a local one which assumes two individuals are similar if they share many common neighbors.
The second type, however, is a global one. It measures the similarity between individuals not based on their direct neighbors, but on their topological structures. This type of similarity measure is also crucial if we want to find people with similar social positions or roles in different communities.
%Clearly, this can not be achieved by simply measure the neighborhood similarity.
As an example, consider the synthetic graph shown in Figure \ref{fig:synthetic}. In terms of neighborhood similarity, vertex $v_0$ is similar to vertices $v_1$, $v_2$ and $v_3$. However, $v_0$ is similar to $v_8$ in terms of structure similarity.

}

\hide{
%and the tie strength with a vertex can be represented by the path similarity value with it.
\begin{figure}[t]
\centering
\includegraphics[width=0.5\textwidth]{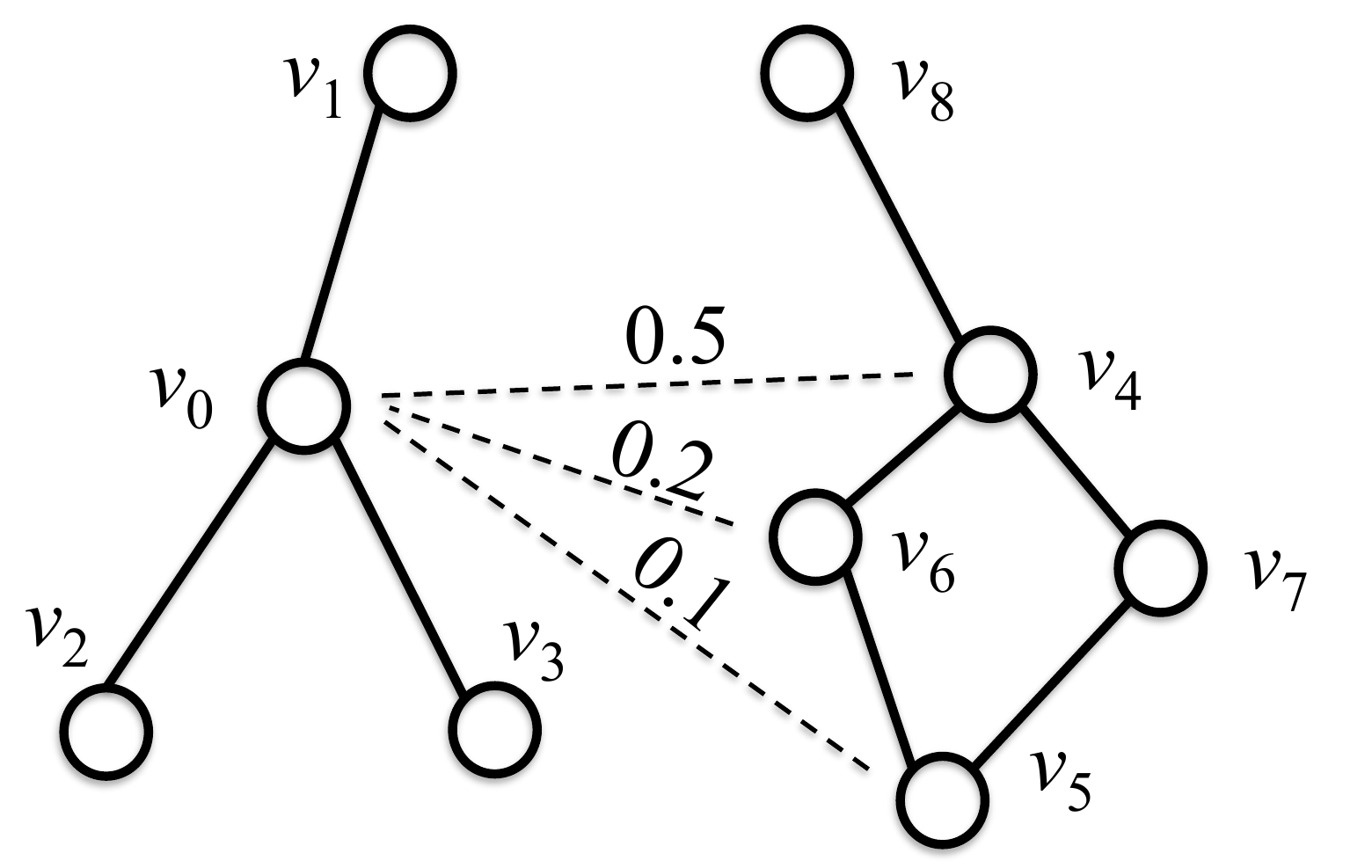}
\caption{\label{fig:synthetic} . A synthetic network.}
\end{figure}

There are two main challenges in this problem:

\vpara{Challenge 1.}
The biggest challenge is its computational efficiency. Existing methods such as SimRank~\cite{jeh:KDD06} and the method proposed in~\cite{leicht:PRE2006} require a time complexity of $O(IN^2d)$, where $I$ is the number of iterations and $d$ is the average degree or the square of the average degree over all vertices. This time complexity is prohibitive when the network becomes larger such as the facebook network, which has over 1 billion vertices.
%Our goal is to find a much more efficient way to solve this problem with provable guarantees on the similarity measure.

\vpara{Challenge 2.}
The second challenge innate to this problem is the definition of structure similarity. RoleSim~\cite{jin:KDD11} incorporates a maximal matching process and the regular SimRank in the hope that this would bias the measure towards structure similarity. It shares the defect of high computational complexity with SimRank. Other methods such as~\cite{henderson:KDD12, henderson:KDD11} deals in an ad-hoc way by defining features recursively for each vertex and then measure the structure similarity based on these features.
%We seek to find a unified way that can employ the neighborhood similarity in measuring the structural similarity.

In this paper, we propose two methods \textit{path similarity}(PathSim)  and \textit{path vector}(PathVec) that correspondingly deal with the neighborhood and structure similarity in a unified framework.

The basic intuition behind PathSim is that ``\textbf{two nodes are similar if they frequently appear on the same path}''. This makes sense if we make an analogy between the path and an information diffusion process. In information diffusion, if two vertices appear in the same diffusion chain, clearly they have similar interests in certain events or objects. As shown in Figure \ref{fig:synthetic}, $v_4$, $v_5$, $v_6$ and $v_7$ frequently co-occur in same paths, they are clearly neighborhood similar.

However, the path similarity is not sufficient to represent the structure similarity between two vertices. Take the synthetic network in Figure \ref{fig:synthetic} as an example. The top similar vertices to vertex $v_0$ is $v_1$, $v_2$, and $v_3$ according to path similarity. Vertex $v_8$ can not be identified although the structures of $v_8$ and $v_0$ are exactly the same. The vertices with long distance or even no connections from the source vertex are easily ignored by path similarity. In terms of Path Vector, we have the idea that \textbf{the probability of a vertex linking to all other vertices are similar if their topology structures are similar}~\cite{holland1981exponential}. We select the top-$k$ path similarities to represent the distribution of tie strength to both improve computational efficiency and capture the most important ties with others. In this way, the structural similarity between two vertices can be measured based on the distance between their path vectors. As shown In Figure \ref{fig:synthetic}, we calculated the path vectors of each vertex. The number outside the parentheses is the path similarity, and the number inside parentheses is the index of the vertex with which the path similarity is calculated. Obviously, the path vectors of $v_8$ and $v_0$ are most similar, and they are indeed structurally similar.

To address the computational problem, we adopt a theoretically-guaranteed sampling based method to approximate the similarity measures. Experiments on both several different datasets have been done to validate the effectiveness and efficiency of our method.

Generally, the contributions of our work include:

\begin{itemize}
\item We novelly define the neighborhood similarity and structural similarity in a unified way. Experiments have shown these similarity measures are general and effective. Our method can indeed find similar vertices in center positions, tight-knit groups, and periphery places in scientific network co-author graph(cf. Figure~\ref{fig:casestudy}).
\item We propose a sampling based method that can approximate the similarities high efficiently. In a network with 1 million vertex and 3 million edges, for the preprocess phase, the required time is less than one minute, and for the on-line query phrase, the required time is millisecond on average (cf. Table~\ref{tb:efficiency}).
\item By adopting the learning theory framework, we give a theoretical justification for the validity of our approximation algorithm, which bridges the gap between formal definition and real practice. These theoretical analyses further guide us to wisely choose parameters used in our algorithm. For instance, given the desired approximation level, the number of random paths is independent of the network size.
\end{itemize}
}

%% file: problem.tex
\section{Problem Formulation}
\label{sec:problem}

We first provide necessary definitions and then formally formulate the problem.

\begin{definition}
\textbf{Undirected Weighted Network.} 
Let $G=(V, E, W)$ denotes a network, where $V$ is a set of $|V|$ vertices and $E\subset V \times V$ is a set of $|E|$ edges between vertices. We use $v_i\in V$ to represent a vertex and 
$e_{ij}\in E$ to represent an edge between vertices $v_i$ and $v_j$. Let $W$ be a weight matrix, with each element $w_{ij} \in W$ representing the weight associated with edge $e_{ij}$. 
\end{definition}

We use $\mathcal{N}(v_i)$ to indicate the set of neighboring vertices of vertex $v_i$. 
We leave the study of directed networks to future work.
Our purpose here is to estimate similarity between two vertices, e.g., $v_i$ and $v_j$.
%In this paper, 
We focus on finding top-$k$ similar vertices. Precisely, the problem can be defined as, given a network $G=(V,E,W)$ and a query vertex $v \in V$, how to find a set $X_{v, k}$ of $k$ vertices that have the highest similarities to vertex $v$, where $k$ is a positive integer. 

A straightforward method to address the top-$k$ similarity search problem is to first calculate the similarity $s(v_i, v_j)$ between vertices $v_i$ and $v_j$  using metrics such as the Jaccard index and SimRank, 
%Let $s(v_i, v_j)$ denotes the similarity between vertices $v_i$ and $v_j$.
and then select a set $X_{v, k}$ of $k$ vertices that have the highest similarities to each vertex $v$.
However it is in general difficult to scale up to large networks.
One important idea is to obtain an approximate set  $X^*_{v, k}$ for each vertex. From the accuracy perspective, we aim to minimize the difference between $X^*_{v, k}$ and $X_{v, k}$. 
Formally, we can define the problem studied in this work as follows.

\begin{problem} \textit{Top-$k$ similarity search.}
Given an undirected weighted network $G=(V, E, W)$, a similarity metric $s(.)$, and a positive integer $k$, any vertex $v \in V$, how to quickly and approximately retrieve the top-$k$ similar vertices of $v$? How to guarantee that the difference  between the two sets $X^*_{v, k}$ and $X_{v, k}$ is less than a threshold $\varepsilon \in (0,1)$, i.e., 
%
% the error of the set $Q(v, k)$ that satisfies 
%
\beq{
%1- \frac{|Q^*(v, k) \cap Q(v, k)|}{|Q^*(v, k) \cup Q(v, k)|} \le \epsilon
\text{Diff}(X^*_{v, k}, X_{v, k}) \le \varepsilon \nonumber
}
\noindent with a probability of at least $1-\delta$. 
%approximate the similarity score $s(v_i, v_j)$? How to minimize the difference between $Q^*(v, k)$ and $Q(v, k)$ (or maximize the similarity between the two sets), i.e.,
\end{problem}

The difference between $X^*_{v, k}$ and $X_{v, k}$ can be also viewed as the error-bound of the approximation. 
In the following section, we will propose a sampling-based method to approximate the top-$k$ vertex similarity. We will explain in details how the method can guarantee the error-bound and how it is able to efficiently achieve the goal.

%% file: approach.tex
\section{Panther: Fast Top-k Similarity Search Using Path Sampling}
\label{sec:approach}

We begin with considering some baseline solutions and then propose our path sampling  approach.
A simple approach to the problem is to consider the number of common neighbors of $v_i$ and $v_j$. If we use the Jaccard index~\cite{Jaccard:1901}, the similarity can be defined as
$$
S_{JA}(v_i, v_j)=\frac{|\mathcal{N}(v_i) \cap \mathcal{N}(v_j)|}{|\mathcal{N}(v_i) \cup \mathcal{N}(v_j)|}. $$
This method only considers local information and does not allow vertices to be similar if they do not share neighbors. 
%Additionally, directly applying this method to find top-$k$ similar vertices  results in a time complexity of $O(|V|d^2)$ with $d$ as the maximal degree, which is still infeasible when handling dense networks.

To leverage the structural information, one can consider algorithms like SimRank~\cite{jeh:KDD06}. SimRank estimates vertex similarity by iteratively propagating vertex similarity to neighbors until convergence (no vertex similarity changes), i.e.,
$$
S_{SR}(v_i, v_j) = \frac{C}{|\mathcal{N}(v_i)||\mathcal{N}(v_j)|} \sum_{v_l\in \mathcal{N}(v_i)} \nonumber
\sum_{v_m \in \mathcal{N}(v_j)} s(v_l, v_m),
$$
\noindent where $C$ is a constant between 0 and 1.
%When defining on a directed network, the neighborhood set $\mathcal{N}(v_i)$ in Eq.~\eqref{eq:simrank} can be replaced by an in-link set (the set of vertices that link to $v_i$).

SimRank similarity depends on the whole network and allows vertices to be similar without sharing neighbors. The problem with SimRank is its high computational complexity: $O(I|V|^2\bar{d}^2)$, which makes it infeasible to scale up to large networks.
 Though quite a few studies have been conducted recently~\cite{He:2010,lee:2012top},
the problem is still largely unsolved.

We propose a sampling-based method to estimate the top-$k$ similar vertices.
In statistics, sampling is a widely used technique to estimate a target distribution~\cite{vapnik1971uniform}. Unlike traditional sampling methods, we propose a random path sampling method, named Panther. Given a network $G=(V, E, W)$, Panther randomly generates $R$ paths with length $T$. Then the similarity estimation between two vertices is cast as estimating how likely it is that two vertices appear on a same path. 
%By building a kd-tree index for the sampled paths, we reduce the complexity to $O(R)$, where $R$ is the number of sampled paths.
%We theoretically prove that with a sample size of $R = O(|E|)$, Panther can achieve an approximation of error $\le \varepsilon$ with a probability of at least $1-\delta$.
Theoretically we prove that given an error-bound, $\varepsilon$, and a confidence level, $1 - \delta$, the sample size $R$ is independent of the network size. Experimentally, we demonstrate that the error-bound is dependent on the number of edges of the network.

\subsection{Random Path Sampling}

%We propose a random path sampling method, named Panther.
The basic idea of the method is that two vertices are similar if they frequently appear on the same paths. The principle is similar to that in Katz~\cite{katz:1953}.
%idea has been also used in previous work such as Katz~\cite{katz:1953}.

\vpara{Path Similarity.}
%We define a new similarity metric, named as \textit{path similarity}. 
To begin with, we introduce how to estimate vertex similarity based on $T$-paths. A $T$-path is defined as a sequence of vertices $p=(v_{1}, \cdots, v_{T+1})$, which consists of $T+1$ vertices and $T$ edges\footnote{Vertices in the same path do not need to be distinct.}. %Without the statement, a path is a sequence with length $T$ henceforth. 
Let $\Pi$ denotes all the $T$-paths in $G$.
Let $w(p)$ be the weight of a path $p$. The weight can be defined in different ways.
Given this, the path similarity between $v_i$ and $v_j$ is defined as:
\beq{
	\label{eq:pathsim}
	S_{RP}(v_i,v_j) =  \frac{\sum_{p \in P_{v_i,v_j}}w(p)}{\sum_{p \in \Pi}w(p)},
}

\noindent where  $P_{v_i,v_j}$ is a subset of $\Pi$ that contain both $v_i$ and $v_j$.

\vpara{Estimating Path Similarity with Random Sampling.}
To calculate the denominator in Eq~\eqref{eq:pathsim}, we need to enumerate all $T$-paths in $G$. 
\hide{
The time complexity is $O(|V|\bar{d}^T)$, where $\bar{d}$ is the average degree and $\bar{d}$ increases fast when the network gets denser. When $\bar{d}$ and $T$ getting large, it becomes intractable to enumerate all the paths. }
%The total number of all paths in $\Pi$ is infinite because we allow the same nodes to appear on the same path more than one time.
However, the time complexity is exponentially proportional to the path length $T$, and thus is inefficient when $T$ increases. 
Therefore, we propose a sampling-based method to estimate the path similarity. The key idea is that we randomly sample $R$ paths from the  network and recalculate Eq~\eqref{eq:pathsim} based on the sampled paths.
\beq{
	\label{eq:pathsim1}
	S_{RP}(v_i,v_j) =  \frac{\sum_{p \in P_{v_i,v_j}}w(p)}{\sum_{p \in P}w(p)}, 
}
\noindent where $P$ is the set of sampled paths.

%In particular, t
To generate a path, we randomly select a vertex in $G$ as the starting point, and then conduct random walks of $T$ steps from $v$ using $t_{ij}$ as the transition probability from vertex $v_i$ to $v_j$.
\beq{
\label{eq:transitionprobability}
t_{ij} = \frac{w_{ij}}{\sum_{v_k \in \mathcal{N}(v_i)} w_{ik}},
}

\noindent where $w_{ij}$ is the weight between  $v_i$ and $v_j$. In a unweighted network, the transition probability can be simplified as $1/|\mathcal{N}(v_i)|$.

Based on the random walk theory~\cite{feller:2008}, we define  $w(p)$ as
$$
	w(p) = \prod_{i=1,j=i+1}^{T}t_{ij}. \nonumber
$$

The path weight also represents the probability that a  path $p$ is sampled from $\Pi$; thus, $w(p)$ in Eq.~\eqref{eq:pathsim1} is absorbed, and we can rewrite the equation as follows:
\beq{
\label{eq:approximatepathsim}
S_{RP}(v_i, v_j) = \frac{|{P}_{v_i,v_j}|}{R}.
}

Algorithm~\ref{algo:gen} summarizes the process for generating the $R$ random paths. %After generating $R$ random paths, 
To calculate Eq.~\eqref{eq:approximatepathsim}, the time complexity is $O(RT)$, because it has to enumerate all $R$ paths. 
To improve the efficiency, we build an inverted index of vertex-to-path~\cite{baeza1999modern}. 
Using the index, we can retrieve all paths that contain a specific vertex $v$ with a complexity of $O(1)$. Then  Eq.~\eqref{eq:approximatepathsim} can be calculated with a complexity of $O(\bar{R}T)$, where $\bar{R}$ is the average number of paths that contain a vertex and $\bar{R}$ is proportional to the average degree $\bar{d}$.
%For an input network, we perform random walk to obtain $R$ paths (Figure~\ref)
%generated random paths and the vertex-to-path index.
Figure~\ref{fig:randompath} illustrates the process of random path sampling. 
Details of the algorithm are presented in Algorithm~\ref{algo:pathsim}.

\begin{figure}[t]
\centering
\includegraphics[width=0.48\textwidth]{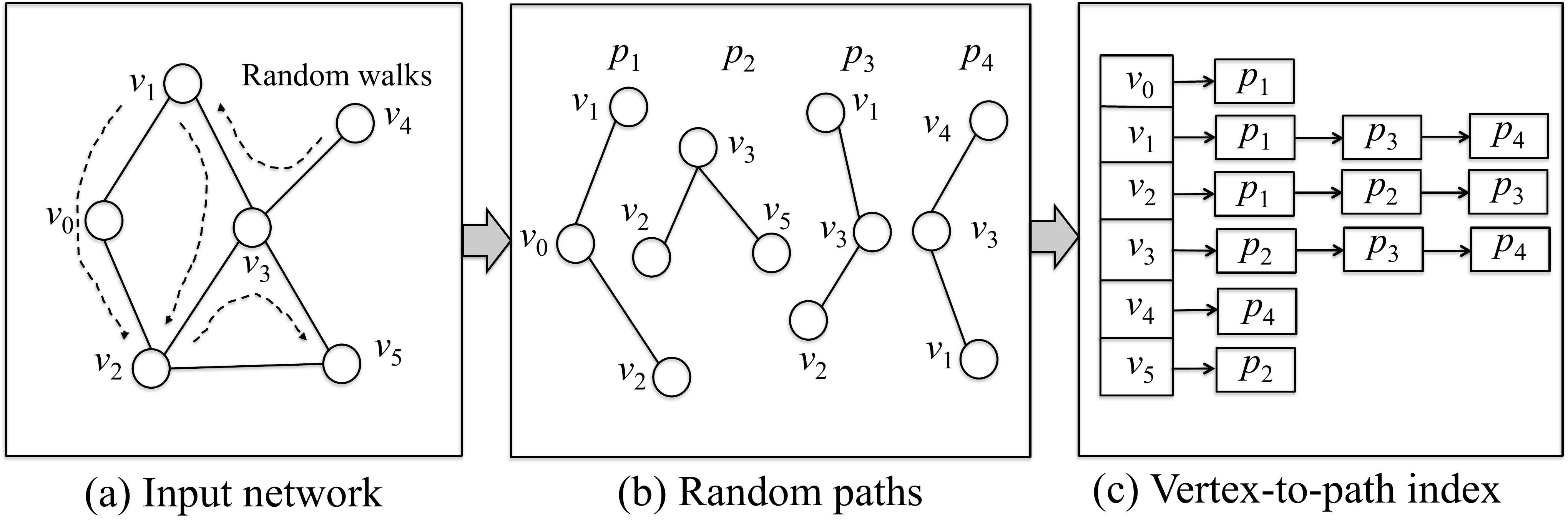}
%\vspace{-0.2in}
\caption{\label{fig:randompath} Illustration of random path sampling.}
\end{figure}

\hide{
Let $G=(V, E, W)$ denotes a network, where $V=\{v_1, \cdots, v_N\}$ is a set of $|V|=N$ vertices and $E\subset V \times V$ is a set of edges between these vertices. W  e use $e_{ij}$ to represent an edge between vertex $v_i$ and vertex $v_j$, and $W$ is the weight matrix with $w_{ij} \in W$ representing the weight associated with the edge $e_{ij}$ The index set of the neighborhood of a vertex $v_i$ is denoted to be $\mathcal{N}(v_i)$.

In this paper, we first consider undirected edges, then we have $e_{ij}=e_{ji}$. We leave directed edges as future study.

We define a path $p$ in a network as a sequence of vertices $(v_1, \cdots, v_{|p|})$, where $(v_i, v_{i+1}) \in E$ for each $ 1\leq i <  |p|$. The length of the path $T$ is defined to be the number of edges in the path, i.e. $T = |p| - 1$. We use $w(p)$ to denote the weight associated with the path $p$. And it is defined to be the product of transition probabilities along the path, that is

\beq{
\label{eq:definitionbak}
	w(p) = \prod_{i=1}^{T}t_{i, i+1}. \nonumber
}

Let $P^T_G$ denote the set of all the paths in $G$ with length $T$, that is $P^T_G= \{ p: |p|-1= T\}$. Let $P^T_{v_i,v_j}$ denote the set of paths in $P^T_{G}$ where two vertices $v_i$ and $v_j$ both appear, i.e. $P^T_{v_i,v_j} = \{p : v_i \in p \text{ and } v_j \in p \text{ and } p \in P^T_{G}\}$. Let $P^T_{v_i}$ denote the set of paths in $P^T_{G}$ where one vertex $v_i$ appears, i.e. $P^T_{v_i} = \{p : v_i \in p \text{ and }  p \in P^T_{G}\}$.

Here we make a formal definition of the path similarity between two vertices $v_i$ and $v_j$ as:
\beq{
	\label{eq:pathsimbak}
	\text{PSim}(v_i,v_j) =  \frac{\sum_{p \in P^T_{v_i,v_j}}w(p)} {\sum_{p \in P^T_{G}}w(p)}.
}

It is intractable to enumerate all the paths with length $T$ where both $v_i$ and $v_j$ appear. Thus we resort to a sampling-based method to approximate the path similarity. The general idea is that we randomly sample $R$ paths from the given network and count the number of co-occurred paths of two vertices as $\#\text{cp}(v_i,v_j)$. We use the ratio between $\#\text{cp}(v_i,v_j)$ and $R$ to approximate the path similarity defined in Eq.~\eqref{eq:definition}. In particular, to generate a path, we first sample a vertex $v$ uniformly at random and then conduct random walks of length $T$ from $v$ to others according to the transition probabilities between vertices.

Details of the algorithm are presented in Algorithm \ref{algo:pathsim}. We leave the justification for this sampling based approximation algorithm for section \ref{sec:proof}.
}

\subsection{Theoretical Analysis}

We give theoretical analysis for the random path sampling algorithm.
%One question prevalent in learning theory community is that how well we can approximate the true distribution using the empirical average. As can be shown later,
In general, the path similarity can be viewed as a probability measure defined over all paths $\Pi$. Thus we can adopt the results from Vapnik-Chernovenkis (VC) learning theory~\cite{vapnik1971uniform} to analyze the proposed sampling-based algorithm.
% by how well we can approximate the true distribution using the empirical average.
To begin with, we will introduce some basic definitions and fundamental results from Vapnik-Chernovenkis theory, and then demonstrate how to utilize these concepts and results to analyze our method.

%\subsection
\vpara{Preliminaries.}
%We outline here some basic definitions and results.
Let $(\mathcal{D}, \mathcal{R})$ be a range space, where $\mathcal{D}$ denotes a domain, and $\mathcal{R}$ is a range set on $\mathcal{D}$.
%, if it is a family of subsets of $\mathcal{D}$. 
%So the elements in $\mathcal{D}$ are points, and the elements in $\mathcal{R}$ are ranges.
 For any set $B \subseteq \mathcal{D}$, $P_{\mathcal{R}}(B) = \{B \cap A : A \in \mathcal{R}\}$ is the projection of $\mathcal{R}$ on $B$.
 If $P_{\mathcal{R}}(B) = 2^{B}$, where $2^{B}$ is the powerset of $B$, we say that the set $B$ is shattered by $\mathcal{R}$. The following definitions and theorem derive from~\cite{Riondato:2014}.
\begin{definition}
 The Vapnik-Chervonenkis (VC) dimension of $\mathcal{R}$, denoted as $VC(\mathcal{R})$, is the maximum cardinality of a subset of $\mathcal{D}$ that can be shattered by $\mathcal{R}$.
\end{definition}

Let $S=\{x_1, \cdots, x_n\}$ be a set of i.i.d. random variables sampled according to a distribution $\phi$ over the domain $\mathcal{D}$. For a set $A \subseteq \mathcal{D}$, let $\phi(A)$ be the probability that an element sampled from $\phi$ belongs to $A$, and let the empirical estimation of $\phi(A)$ on $S$ be
$$
 	\phi_{S}(A) = \frac{1}{n} \sum_{i=1}^{n} \mathds{1}_A(x_i), \nonumber
 $$
 \noindent where $\mathds{1}_A$ is the indicator function with the value of $\mathds{1}_A(x)$ equals 1 if $X \in A$, and 0 otherwise.

The question of interest is that how well we can estimate $\phi(A)$ using its unbiased estimator, the empirical estimation $\phi_{S}(A)$. We first give the goodness of approximation in the following definition.
\begin{definition}
 Let $\mathcal{R}$ be a range set on $\mathcal{D}$, and $\phi$ be a probability distribution defined on $\mathcal{D}$. For $\varepsilon \in (0,1)$, an $\varepsilon$-approximation to $(\mathcal{R}, \phi)$ is a set $S$ of elements in $\mathcal{D}$ such that
$$
    sup_{A \in \mathcal{R}}|\phi(A) - \phi_S(A)| \leq \varepsilon.\nonumber
 $$
\end{definition}

One important result of VC theory is that if we can bound the $VC$-dimension of $\mathcal{R}$, it is possible to build an $\varepsilon$-approximation by randomly sampling points from the domain according to the distribution $\phi$. This is summarized in the following theorem.

\begin{theorem} \label{eq:samplesize}
Let $\mathcal{R}$ be a range set on a domain $\mathcal{D}$, with $VC(\mathcal{R}) \leq d$, and let $\phi$ be a distribution on $\mathcal{D}$. Given $\varepsilon, \delta \in (0,1)$ , let $S$ be a set of $|S|$ points sampled from $\mathcal{D}$ according to $\phi$, with
$$
	|S| =  \frac{c}{\varepsilon^2} (d +\ln \frac{1}{\delta}), \nonumber
$$
\noindent where $c$ is a universal positive constant. Then $S$ is a $\varepsilon$-approximation to $(\mathcal{R}, \phi)$ with probability of at least $1-\delta$.
\end{theorem}

\para{Range Set of Path.}
In our setting, we set the domain to be $\Pi$---the set of all paths with length $T$ in the graph $G$. Accordingly, we define the range set $\mathcal{R}_G$ on domain $\Pi$ to be
$$\mathcal{R}_G = \{ P_{v_i, v_j} : v_i, v_j \in V\}. $$

It is a valid range set, since it is the collection of subsets $P_{v_i, v_j}$ of domain $\Pi$. We  first show an upper bound of the VC dimension of $\mathcal{R}_G$ in Lemma \ref{eq:vc}.
%The range set  is defined on the set $P^T_G$ of all the paths with length $T$ in graph $G$. It contains, for any two vertices $v_i, v_j \in V$, the set $P_{v_i,v_j}^T$ of paths with length $T$ that $v_i$ and $v_j$ both appear:
The proof is inspired by ~\cite{Riondato:2014}.
\begin{lemma} \label{eq:vc} $VC(\mathcal{R}_G)  \leq   \log_2 \binom{T}{2} +1$
\end{lemma}

\begin{proof}
We  prove the lemma by contradiction. Assume  $VC(\mathcal{R}_G) = l$ and  $l >  \log_2 \binom{T}{2} +1$.
By the definition of VC-dimension, there is a set $Q  \subseteq \Pi$ of size $l$ that can be shattered by $\mathcal{R}_G$. That is, we have the following statement:
$$
\forall S_i \subseteq Q \text{ , } \exists P_i\in \mathcal{R}_G \text{, ~~~~s.t. } P_i\cap Q = S_i,  $$

\noindent where $P_i$ is the $i$-th range.
\hide{From the above definition, we can get that any two paths $p_1, p_2 \in Q$, we have neither $p_1 \subseteq p_2$ nor $p_2 \subseteq p_1$. otherwise one of the two paths would appear in all ranges where the other one appears, and so it would be impossible to shatter $Q$.  We prove this by contradiction.  We assume there are two paths $p_1, p_2 \in Q$ and $p_1 \subseteq p_2$, e.g., $p_1=\{v_1,v_2\}$ and $p_2 = \{v_1,v_2\}$. Then we can only find $P_{v_1,v_2} = \{p_1,p_2\}$ such that $P_{v_1,v_2}  \cap Q =\{p_1, p_2\}$, but cannot find any $P_{v_i, v_j} \in \mathcal{R}_G$ such that $P_{v_i, v_j}  \cap Q = \{p_1\}$ or $P_{v_i, v_j}  \cap Q = \{p_2\}$,  we reach the contradiction that $Q$ can be shattered by $\mathcal{R}_G$.
}
Since  each subset $S_i \in Q$ is different from the other subsets, the corresponding range $P_i$  that making  $P_i\cap Q = S_i$ is also different from the other ranges. Moreover, the set $Q$ is shattered by $\mathcal{R}_G$ if and only if $\{P_i\cap Q : P_i \in \mathcal{R}\} = 2^Q$. Thus  $\forall p \in Q$, there are $2^{l-1}$ non-empty distinct subsets $S_1, \cdots, S_{2^{l-1}}$ of $Q$ containing the path $p$. 
%Given that $Q$ is shattered by $\mathcal{R}_G$, for each set $S_i$,  there exists a range $P_i \in \mathcal{R}_G$ such that $S_i = Q \cap P_i$.  
%Since all the subsets $S_i$'s are different from each other, then all the $P_i$'s are also different from each other. 
So there are also $2^{l-1}$ distinct ranges in $\mathcal{R}_G$ that contain the path $p$, 
%Given that $p$ is a member of every $S_i$, $p$ must also belong to each $P_i$; that is, there are $2^{l-1}$ distinct ranges in $\mathcal{R}_G$ containing $p$, 
i.e.
$$	|\{ P_i| p \in P_i \text{ and } P_i  \in \mathcal{R}_G\}| = 2^{l-1}.  $$

In addition, according to the definition of range set, $\mathcal{R}_G = \{ P_{v_i, v_j} : v_i, v_j \in V\}$, we know that a path belongs only to the ranges corresponding to any pair of vertices in path $p$, i.e., to the pairwise combinations of the vertices in $p$. This means the number of ranges in $\mathcal{R}_G$ that $p$ belongs to is equal to the combinatorial number $\binom{T}{2} $, i.e.,
$$
		|\{ P_i| p \in P_i \text{ and } P_i  \in \mathcal{R}_G\}| = \binom{T}{2} . \nonumber
$$

On the other hand, from our preliminary assumption, we have $l >  \log_2 \binom{T}{2} +1$, which is equivalent to $\binom{T}{2}  < 2^{l-1}$. Thus,
$$
	|\{ P_i| p \in P_i \text{ and } P_i  \in \mathcal{R}_G\}| = \binom{T}{2}  < 2^{l-1}. \nonumber
$$

Hence, we reach a contradiction: it is impossible to have $2^{l-1}$ distinct ranges $P_i \in \mathcal{R}_G$ containing $p$.
Since there is a one-to-one correspondence between $S_i$ and $P_i$, we get that it is also impossible to have $2^{l-1}$ distinct subset $S_i \in Q$ containing $p$.  Therefore, we prove that $Q$ cannot be shattered by $\mathcal{R}_G$ and $VC(\mathcal{R}_G)  \leq   \log_2 \binom{T}{2}  +1.$
\end{proof}

%\subsection
\para{Sample Size Guarantee.}
We now provide theoretical guarantee for the number of sampled paths.
%In another word, h
How many random paths do we need to achieve an error-bound $\varepsilon$ with probability $1-\delta$?
We define a probability distribution on the domain $\Pi$. $\forall p \in \Pi$, we define
$$
	\phi(p) = \text{prob}(p) = \frac{w(p)}{\sum_{p \in \Pi}w(p)}.
$$

We can see that the definition of $S_{RP}(v_i, v_j)$ in Eq.\eqref{eq:pathsim} is equivalent to $\phi(P_{v_i, v_j})$. This observation enables us to use a sampling-based method (empirical average) to estimate the original path similarity (true probability measure).

Plugging the result of Lemma \eqref{eq:vc} into Theorem \eqref{eq:samplesize}, we obtain:
$$
	R = \frac{c}{\varepsilon^2} ( \log_2 \binom{T}{2} +1 +\ln \frac{1}{\delta}). 
$$

That is, with at least $R$ random paths, we can estimate the path similarity between any two vertices with the desired error-bound and confidence level. The above equation also implies that the sample size $R$ only depends on the path length $T$, given an error-bound $\varepsilon$, and a confidence level $1-\delta$.

\begin{algorithm}[!t]
     \caption{Panther}\label{algo:pathsim}

       \KwIn{A network $G=(V, E, W)$, path length $T$,  parameters $\varepsilon$, $c$, $\delta$, a vertex $v$, and $k$.}
       \KwOut{top-$k$ similar vertices with regard to $v$.}
       Calculate sample size $R = \frac{c}{\varepsilon^2} ( \log_2 \binom{T}{2} +1 +\ln \frac{1}{\delta})$ \;
       GenerateRandomPath($G$, $R$)\;
       	\ForEach{ $p_n \in P_{v}$}{
       			\ForEach{ \text{Unique} $v_j \in p_n$}{
       				$S_{RP}(v, v_j) +=\frac{1}{R}$ \;
       			}
       	}
       	Retrieve top-$k$ similar vertices according to $S_{RP}(v, v_j)$\;
\end{algorithm}

\begin{algorithm}[!t]
     \caption{Panther++}\label{algo:pathvec}

       \KwIn{A network $G=(V, E, W)$, path length $T$,  parameters $\varepsilon$, $c$, $\delta$, vector dimension $D$, a vertex $v$, and $k$.}
       \KwOut{top-$k$ similar vertices with regard to $v$.}
       Calculate sample size $R = \frac{c}{\varepsilon^2} ( \log_2 \binom{T}{2} +1 +\ln \frac{1}{\delta})$ \;
       GenerateRandomPath($G$, $R$)\;
       \ForEach{ $v_i \in V$ }{
       		\ForEach{ $p_n \in P_{v_i}$}{
       			\ForEach{ \text{Unique} $v_j \in p_n$}{
       					$S_{RP}(v_i, v_j) +=\frac{1}{R}$ \;
       			}
       		}
       		Construct a vector $\theta(v_i)$ by taking the largest $D$ values from $\{S_{RP}(v_i, v_j): v_j \in p_n  \text{ and }  p_n \in P_{v_i}\}$\;
       }
       Build a kd-tree index based on the Euclidean distance between any vectors $\theta(v_i)$ and $\theta(v_j)$ \;
       Query the top-$k$ similar vertices from the index for  $v$\;
\end{algorithm}

\begin{algorithm}[!t]
     \caption{GenerateRandomPath}\label{algo:gen}

       \KwIn{A network $G=(V,E, W)$ and sample size $R$.}
       \KwOut{Paths $\{p_r\}_{r=1}^{R}$ and vertex-to-path index $\{P_{v_i}\}_{i=1}^N$.}

       Calculate transition probabilities between every pair of vertices according to Eq.~\eqref{eq:transitionprobability} \;
       Initialize $r=1$\;
       \Repeat{$r < R$}
       {
       		Sample current vertex $v = v_i$ uniformly at random \;
       	    Add $v$ into $p_r$ and add $p_r$ into the path set of $v$, i.e., $P_{v}$ \;
       	    \Repeat{$|p_r| < T+1$}
       	    {
       	    	 Randomly sample a neighbor $v_j$ according to transition probabilities from $v$ to its neighbors\;
       	    	 Set current vertex $v = v_j$\;
       	    	 Add $v$ into $p_r$ and add $p_r$ into $P_{v}$ \;
       	    }
       		$r+=1$\;
       }
\end{algorithm}

\subsection{Panther++}
One  limitation of \model is that the similarities obtained by the algorithm  have a bias to close neighbors, though in principle it considers the structural information. We therefore present an extension of the Panther algorithm. The idea is to augment each vertex with a feature vector. To construct the feature vector, we follow the intuition that the probability of a vertex linking to all other vertices is similar if their topology structures are similar~\cite{holland1981exponential}. We select the top-$D$ similarities calculated by Panther to represent the probability distribution.
Specifically, for vertex $v_i$ in the network, we first calculate the similarity between $v_i$ and all the other vertices using Panther. Then we construct a feature vector for $v_i$ by taking the largest $D$ similarity scores as feature values, i.e.,
$$
\theta(v_i) = (S_{RP}(v_i, v_{(1)}), S_{RP}(v_i, v_{(2)}), \ldots , S_{RP}(v_i, v_{(D)})), $$

\noindent where $S_{RP}(v_i, v_{(d)})$ denotes the $d$-th largest path similarity between $v_i$ and another vertex $v_{(d)}$.

Finally, the similarity between $v_i$ and $v_j$ is re-calculated as the reciprocal Euclidean distance between their feature vectors:
$$
	S_{RP^{++}}(v_i,v_j) =   \frac{1}{  \lVert \theta(v_i) - \theta(v_j)\rVert}. \nonumber
$$

The idea of using vertex features to estimate vertex similarity was also used for graph mining~\cite{henderson:KDD11}.

\vpara{Index of Feature Vectors}
Again, we use the indexing techniques to improve the algorithm efficiency.
We build a memory based kd-tree~\cite{wald2006building} index for feature vectors of all vertices. 
Then given a vertex, we can retrieve top-$k$ vertices  in
the kd-tree with the least Euclidean distance to the query vertex efficiently. 
At a high level, a kd-tree is a generalization of a binary search tree that stores points in $D$-dimensional space. In level $h$ of a kd-tree, given a node $v$, the $h \% D$-th element in the vector of each node in its left subtree is less than the $h \% D$-th element in the vector of $v$, while the $h \% D$-th element of every node in the right subtree is no less than the $h \% D$-th element of $v$. Figure~\ref{fig:kd-tree} shows the data structure of the index built in Panther++. Based on the index, we can  query whether a given point is stored in the index very fast. Specifically, given a vertex $v$, if the root node is $v$, return the root node. If the first element of $v$ is strictly less than the first element of the root node, look for $v$ in the left subtree, then compare it to the second element of $v$. Otherwise, check the right subtree.
It is worth noting that we can easily replace kd-tree with any other index methods, such as r-tree.
The algorithms for calculating feature vectors of all vertices and the similarity between vertices are shown in Algorithm \ref{algo:pathvec}.

\begin{figure}[t]
\centering
\includegraphics[width=0.4\textwidth]{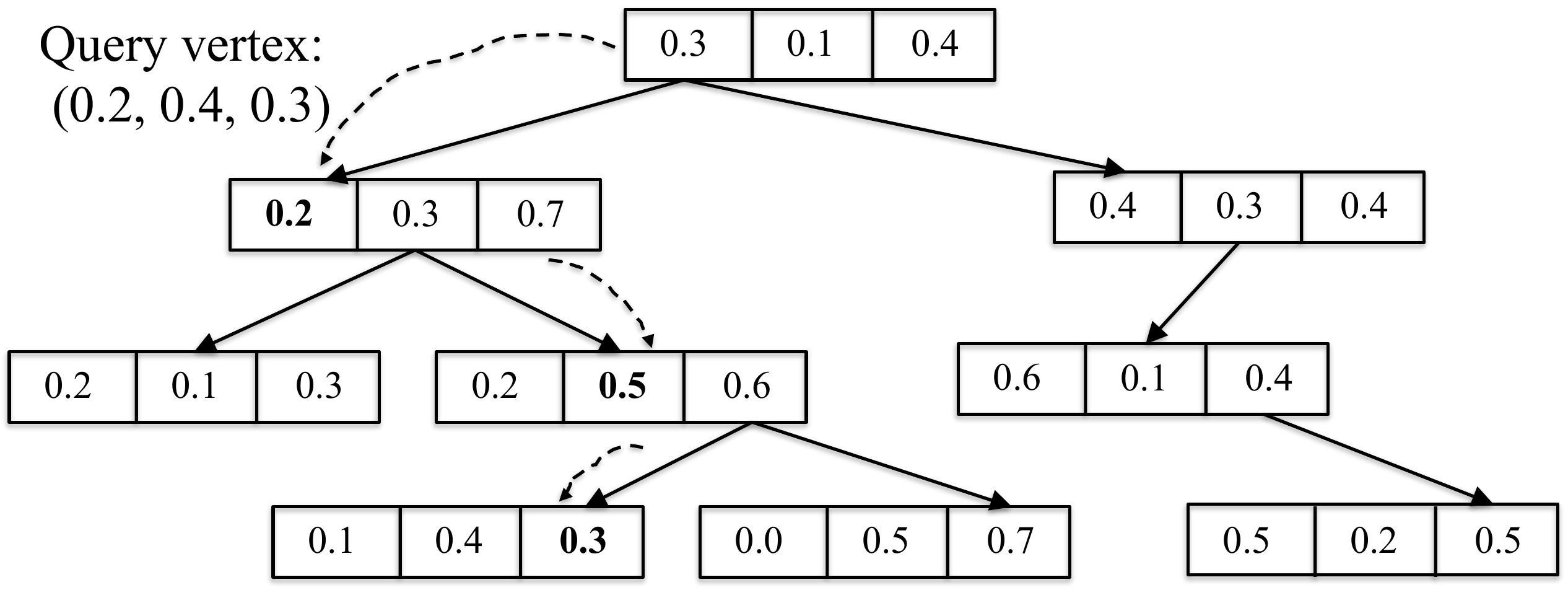}
\caption{\label{fig:kd-tree} Data structure of the index built in Panther++.}
\end{figure}

%\subsection
\vpara{Implementation Notes.}
In our experiments, we empirically set the parameters as follows:  $c=0.5$, $\delta=0.1$,  $T=5$, $D=50$ and $\varepsilon=\sqrt{1/|E|}$. The optimal values of $T$, $D$ and $\varepsilon$ are discussed in section \ref{sec:exp}. We build the kd-tree using the toolkit ANN\footnote{http://www.cs.umd.edu/~mount/ANN/}.
%Our method is easy to be parallelized because the process of generating random paths and constructing feature vector for a vertex is independent of that for other vertices. Thus we also implement the multi-thread version of our method, which parallelizes the outmost loops in Algorithm \ref{algo:gen} and Algorithm\ref{algo:pathvec} respectively.

\begin{table}
\centering
\caption{\label{tb:current_method}  Time and space complexity for calculating top-$k$ similar vertices for all vertices in a network. $I$--- number of iterations, $\bar{d}$---average degree, $f$---feature number, $D$ --- vector dimension, and $T$ --- path length. }
\renewcommand\arraystretch{1.2}
\small 
\begin{tabular}{@{ }c@{ }|c@{ }|c@{}}
\hline
\textbf{Method}&\textbf{Time Complexity}&\textbf{Space Complexity} \\
\hline
SimRank~\cite{jeh:KDD06}                   &  $O(I|V|^2\bar{d}^2)$          &  $O(|V|^2)$ \\
TopSim~\cite{lee:2012top}                   &  $O(|V|T\bar{d}^T)$             &  $O(|V|+|E|)$ \\
RWR~\cite{pan2004automatic}             &  $O(I|V|^2\bar{d})$              &  $O(|V|^2)$ \\
RoleSim~\cite{jin:KDD11}                     &  $O(I|V|^2\bar{d}^2)$          &  $O(|V|^2)$ \\
ReFex~\cite{henderson:KDD11}            &   $O(|V|+I(f|E| + |V|f^2))$             & $O(|V|+|E|f)$  \\
Panther                                              &  $O(RTc + |V| \bar{d} T)$            & $O(RT + |V| \bar{d})$  \\
Panther++                                         &   $O(RTc +|V| \bar{d} T  +|V|c)$            & $O(RT+|V|\bar{d} + |V|D)$  \\
\hline
\end{tabular}
\end{table}

\subsection{Complexity Analysis}
In general, existing methods result in high complexities.
%take a long time to find top-$k$ similar vertices for any vertex in $G$. 
For example, the time complexity of SimRank~\cite{jeh:KDD06}, TopSim~\cite{lee:2012top}, Random walk with restart (RWR)~\cite{pan2004automatic}, RoleSim~\cite{jin:KDD11}, and ReFex~\cite{henderson:KDD11} is $O(I|V|^2\bar{d}^2)$, $O(|V|T\bar{d}^T)$,  $O(I|V|^2\bar{d})$, $O(I|V|^2\bar{d}^2)$, and  $O(|V|+I(f|E| + |V|f^2))$, respectively. 
Table~\ref{tb:current_method} summarizes the time and space complexities of the different  methods.
For Panther, its time complexity  includes two parts:
\begin{itemize}
\item \textbf{Random path sampling:} The time complexity of generating random paths is $O(RT \log \bar{d})$, where $\log \bar{d}$ is very small and can be simplified as a small constant $c$. Hence, the time complexity is $O(RTc)$.
\item \textbf{Top-$k$ similarity search:} The time complexity of calculating top-$k$ similar vertices for all vertices is $O( |V|\bar{R}T + |V|\bar{M} )$. The first part $O( |V|\bar{R}T)$ is the time complexity of calculating Eq.~\eqref{eq:approximatepathsim} for all pairs of vertices, where $\bar{R}$ is the average number of paths that contain a vertex and is proportional to the average degree $\bar{d}$. The second part $O( |V| \bar{M} )$ is the time complexity of searching top-$k$ similar vertices based on a  heap structure, where $\bar{M}$ represents the average number of co-occurred vertices with a vertex and is proportional to $\bar{d}$. Hence, the time complexity is $O(|V| \bar{d} T)$.
 \end{itemize}
 
 The space complexity for storing paths and vertex-to-path index is $O(RT)$  and $O(|V|\bar{d})$, respectively.  
 
%Based on Panther, 
Panther++ requires additional computation to build the kd-tree. The time complexity of building a  kd-tree is $O(|V|\log |V|)$ and querying top-$k$ similar vertices for any vertex is $O(|V| \log |V|) $, where $\log |V|$ is  small and can be viewed as a small constant $c$. Additional space (with a complexity of $O(|V|D)$) is required to store $|V|$ vectors with $D$-dimension.

%% file: exp.tex
\section{Experiments}
\label{sec:exp}

\subsection{Experimental Setup}

In this section, we conduct various experiments to evaluate the proposed methods for top-$k$ similarity search.

\hide{
\begin{itemize}
\item Does PathSim successfully capture the similarity between vertices with many common neighbors?
\item Does PathVec capture the similarity between vertices with similar structures?
\item How well do PathSim and PathVec perform in terms of computational efficiency?

\end{itemize}

}

\vpara{Datasets.} We evaluate the proposed method on four different networks: Tencent, 
Twitter, Mobile, and co-author. 

\textbf{Tencent~\cite{yang2014rain}}: The dataset is from Tencent Weibo\footnote{http://t.qq.com}, a popular Twitter-like microblogging service in China, and consists of over 355,591,065 users and 5,958,853,072 ``following'' relationships. 
%If there exists a `` following'' link from a user $v$ to another user $u$, we build a link between them. 
The weight associated with each edge is set as 1.0 uniformly. 
This is the largest network in our experiments. 
We mainly use it to evaluate the efficiency performance of our methods.
% in Tencent network.
%We sample sub-networks with different scales from the original network by randomly selecting vertices and employing breadth-first search.

\textbf{Twitter~\cite{Hopcroft:11CIKM}}: The dataset was crawled in the following way. %To begin the collection process, w
We first selected the most popular user on Twitter, i.e., ``Lady Gaga'', and randomly selected 10,000 of her followers. We then collected all followers of these users. In the end, we obtained 113,044 users and 468,238 ``following'' relationships in total. The weight associated with each edge is also set as 1.0 uniformly. We use this dataset to evaluate 
the accuracy of Panther and Panther++. % in Twitter network.

\textbf{Mobile~\cite{Dong:14KDD}}: The dataset is from a mobile communication company, and consists of millions of call records.
Each call record contains information about the sender, the receiver, the starting time, and the ending time. We build a network using call records within two weeks by treating each user as a vertex, and communication between users as an edge.
%if they make calls with each other. 
The resultant network consists of 194,526 vertices and 206,934 edges.
% in total. 
The weight associated with each edge is defined as the number of calls. We also use this dataset to evaluate the accuracy of the proposed methods.

\textbf{Co-author~\cite{Tang:08KDD}}: The dataset\footnote{http://aminer.org/citation} is from AMiner.org, and contains 2,092,356 papers. From the original citation data, we extracted a weighted co-author graph from each of the following conferences from  2005 to 2013: KDD, ICDM, SIGIR, CIKM, SIGMOD, ICDE, and ICML\footnote{  Numbers of vertices/edges of different conferences are: KDD: 2,867/ 7,637, ICDM: 2,607/4,774, SIGIR: 2,851/6,354, CIKM: 3,548/7,076, SIGMOD: 2,616/8,304, ICDE: 2,559/6,668, and ICML: 3511/6105.}. The weight associated with each edge is the number of papers collaborated on by the two connected authors. We also use the dataset to evaluate the accuracy of the proposed methods.

%We treat all the networks as undirected networks in the experiments.

\hide{
\begin{table*}
\centering
\caption{\label{tb:dataset} Datasets. ($|V|$ is the number of vertices, $|E|$ is the number of edges.)}
\begin{tabular}{c|c|c|c|c|c|c|c|c|c}
\hline
Scale  &  KDD  & ICDM  &  SIGIR  & CIKM & SIGMOD & ICDE &  Twitter & Mobile & Science\\
\hline
$|V|$ & 2,867  & 2,607   &  2,851  & 3,548 & 2,616      & 2,559 & 112,044 &  194,526 & 1,589   \\
$|E|$ & 7,637  & 4,774   &  6,354  & 7,076 & 8,304      & 6,668 & 468,238  &  206,934 & 2,743 \\
\hline
\end{tabular}
\end{table*}
}

\vpara{Evaluation Aspects.} To quantitatively evaluate the proposed methods, we consider the following performance measurements:

\textbf{Efficiency Performance:} We apply our methods to the Tencent network to evaluate the computational time.
%as the efficiency metric.

\textbf{Accuracy Performance:} We apply the proposed methods to recognize identical authors on different co-author networks. 
We also apply our methods to the Coauthor, Twitter and Mobile networks to 
evaluate how they estimate the top-$k$ similarity search results.
%for finding top-$k$ similar vertices and evaluate whether the top similar vertices own numerous common neighbors or same social positions. 

\textbf{Parameter Sensitivity Analysis:} We analyze the sensitivity of different parameters in our methods: path length $T$, vector dimension $D$, and error-bound $\varepsilon$.
% with regards to the accuracy of Panther and Panther++.

Finally, we also use several case studies as anecdotal evidence to further demonstrate the effectiveness of the proposed method.
All codes are implemented in C++ and compiled using GCC 4.8.2 with -O3 flag.
The experiments were conducted on a Ubuntu server with four Intel Xeon(R) CPU E5-4650 (2.70GHz) and 1T RAM.

\vpara{Comparison methods.} We compare with the following methods:

\textbf{RWR~\cite{pan2004automatic}}: Starts from $v_i$, iteratively walks to its neighbors with the probability proportional to their edge weights. At each step, it also has some probability to walk back to $v_i$ (set as 0.1). The similarity between $v_i$ and $v_j$ is defined as the steady-state probability that $v_i$ will finally reach at $v_j$. We calculate RWR scores between all pairs and then search the top-$k$ similar vertices for each vertex.

\textbf{TopSim~\cite{lee:2012top}}: Extends SimRank~\cite{jeh:KDD06} on one graph $G$ to finding top-$k$ authoritative vertices on  the product graph $G \times G$ efficiently.

\textbf{RoleSim~\cite{jin:KDD11}}: Refines SimRank~\cite{jeh:KDD06} by changing the average similarity of all neighbor pairs to all matched neighbor pairs. We calculate RoleSim scores between all pairs and then search the top-$k$ similar vertices for each vertex.
\hide{
\beq{\small
	S(v_i,v_j) = (1-\beta) \max_{\mathcal{M}(v_i,v_j)} \frac{\sum_{(v_l, v_m) \in \mathcal{M}(v_i, v_j)}s(v_l, v_m)}{|\mathcal{N}(v_i)| +|\mathcal{N}(v_j)| -|\mathcal{M}(v_i, v_j)|} + \beta. \nonumber
\normalsize
}
\noindent where $\mathcal{M}(u,v)$ is a matching between $\mathcal{N}(u)$ and $\mathcal{N}(v)$, and $\beta$ is a decay factor. 
%and is set as 0.1 as in the original paper. 
%
}

\textbf{ReFeX~\cite{henderson:KDD11}}: Defines local, egonet, and recursive features to capture the structural characteristic. Local feature is the vertex degree. Egonet features include the number of within-egonet edges and the number of out-egonet edges. For weighted networks, they contain weighted versions of each feature. Recursive features are defined as the mean and sum value of each local or egonet feature among all neighbors of a vertex. In our experiments, we only extract recursive features once and construct a vector for each vertex by a total of 18 features.  For fair comparison, to search top-$k$ similar vertices, we also build the same kd-tree as that in our method.

The codes of TopSim, RoleSim, and ReFex are provided by the authors of the original papers. We tried to use the fast versions of TopSim and RoleSim mentioned in their paper.

\begin{table*}
 \centering
\caption{\label{tb:efficiency}Efficiency performance (CPU time) of comparison methods on different sizes of the Tencent sub-networks. 
The time includes all computational cost for processing and top-$k$ similarity search for all vertices.  The time before ``+'' denotes the time used for processing and the time after ``+'' denotes that used for top-$k$ similarity search. ``---'' indicates that the corresponding algorithm cannot finish the computation within a reasonable time.}
%CPU time costs of different methods (Pre-computational time + On-line query time).
\renewcommand\arraystretch{1.1}
%\small
\begin{tabular}{@{}l@{}|l|l|l|l|l|l|l|l@{}}
\hline 
\textbf{Sub-network} &  $\textbf{|V|}$ &
\textbf{|E| } & \textbf{RWR} &\textbf{TopSim}&\textbf{RoleSim} & \textbf{ReFeX} &\textbf{Panther}  & \textbf{Panther++} \\
\hline
\hline
Tencent1 	&	 6,523 & 10,000        &  +7.79hr & +28.58m          &  +37.26s                       &  3.85s+0.07s         &  0.07s+0.26s   &  0.99s+0.21s   \\
Tencent2 	&	25,844	   & 50,000    &  +>150hr   &+11.20hr           &  +12.98m                    &  26.09s+0.40s         &   0.28s+1.53s    &  2.45s+4.21s\\
Tencent3 	&	48,837 	 & 100,000      &  ---  &+30.94hr                        &  +1.06hr           &  2.02m+0.57s     &   0.58s+ 3.48s     &  5.30s+5.96s \\ 
Tencent4 	&	169,209	  &   500,000 &--- 		    &  +>120hr                    &  +>72 hr         &  17.18m+2.51s&   8.19s+16.08s  &  27.94s+24.17s  \\  
Tencent5    	& 230,103 	   &  1,000,000   &---       &  ---                      &  ---                     &  31.50m+3.29s&   15.31s+30.63s  &  49.83s+22.86s\\
Tencent6      &   443,070    &   5,000,000     &---     &  ---                      &  ---            &  24.15hr+8.55s         &   50.91s+2.82m  &  4.01m+1.29m\\
Tencent7 	  &  702,049     &  10,000,000     &---   &  ---                     &  ---               &>48hr                        &  2.21m+6.24m    & 8.60m+6.58m\\
Tencent8 	 & 2,767,344   &   50,000,000     &---   &  ---                      &  ---                                 &  ---      & 15.78m+1.36hr   &  1.60hr+2.17hr  \\
Tencent9 	 & 5,355,507  &   100,000,000     & ---&  ---                      &  ---                                 &  ---      &  44.09m +4.50hr &  5.61hr +6.47hr \\
Tencent10 	 & 26,033,969  &   500,000,000   &---   &  ---                      &  ---                                 &  ---      &  4.82hr +25.01hr &  32.90hr +47.34hr \\
Tencent11 	 & 51,640,620  &   1,000,000,000  &---    &  ---                      &  ---                                 &  ---     &  13.32hr +80.38hr &  98.15hr +120.01hr \\

\hline
\end{tabular}
\end{table*}

\subsection{Efficiency and Scalability Performance}

\hide{
\begin{table}[t]
\centering
\caption{\label{tb:tencentdataset} Tencent sub-networks. ($|V|$ is the number of vertices, $|E|$ is the number of edges ($\times 10^4$).)}
\begin{tabular}{ c|c|c|c|c|c}
\hline
Data set &  $|V|$  &  $|E|$  & Data set & $|V|$ & $|E|$    \\
\hline
Tencent1 	&	 6,523  	&     1 	  &  Tencent5    	& 230,103 	   &  100  \\
Tencent2 	&	25,844	   &	5        &  Tencent6      &   443,070    &   500   \\
Tencent3 	&	48,837 	   &	10      &  Tencent7 	  &  702,049     &   1000   \\
Tencent4 	&	169,209	  &    50	   &  Tencent8 	 & 2,767,344   &   5000  \\

\hline
\end{tabular}
\end{table} 
}

In this subsection, we first fix $k=5$, and evaluate the efficiency and scalability performance of different comparison methods using the Tencent dataset. We evaluate the performance by randomly extracting different (large and small) versions of the Tencent networks. %Table~\ref{tb:tencentdataset} lists statistics of the different Tencent sub-networks.
%Given the task of finding the top-5 similar vertices for each vertex in the network, we record the time of each method in the sampled Tencent sub-networks with different scales (Table in Figure~\ref{fig:efficiency}).
For TopSim and RoleSim, we only show the computational time for similarity search. For ReFex, Panther, and Panther++, we also show 
%record the time of both 
the computational time used for preprocessing.
% part and the total time of on-line query of all the vertices. 

%only has the on-line query stage and RoleSim only has the pre-computational stage.

%\vpara{Efficiency.}
Table~\ref{tb:efficiency} lists statistics of the different Tencent sub-networks and the efficiency performance of the comparison methods.
Clearly, our methods (both \model and Panther++)  are much faster than the comparison methods. For example, on the Tencent6 sub-network, which consists of 443,070 vertices and 5,000,000 edges, 
\model achieves a $390\times$ speed-up , compared to the fastest (ReFeX) of all the comparative methods.

Figure~\ref{subfig:speeduprefex} shows the speed-up of Panther++ compared to ReFeX on different scales of sub-networks. The speed-up is moderate when the size of the network is small ($|E| \le 1,000,000$); when continuing to increase the size of the network, the obtained speed-up is even superlinear. 
We conducted a result comparison between ReFeX and Panther++. 
The results of Panther++ are very similar to those of ReFex, though they decrease slightly when the size of the network is small. Figure~\ref{subfig:parameterk} shows the efficiency performance of Panther and Panther++ by varying the values of $k$ from 5 to 100. We can see that the time costs of Panther and Panther++ are not very sensitive to $k$. The growth of time cost is slow  when $k$ gets larger. This is because $k$ is only related to the time complexity of top-$k$ similarity search based on a heap structure. When $k$ gets larger, the time complexity approximates to $O(\bar{M} \log \bar{M})$ from $O(\bar{M})$, where $\bar{M}$ is the average number of co-occurred vertices on the same paths. We can also see that the time cost is not very stable when $k$ gets larger, because the paths are randomly generated, which results in different values of $\bar{M}$ each time.%From the results shown in Figure~\ref{fig:efficiency}, e can see that in the network consisting of 230,103 vertices and 1,000,000 edges, our method Panther++ achieves a XX$\times$ speedup than the fastest comparative methods. 
%We also present the ratio of the scores evaluated on the ground truth of structural holes between ReFex and Panther++ (see Section \ref{sec:accuracy} for details), and the ratio of time costs between ReFex and Panther++ in the networks with $|E|$ from $10^4$ to $500 \times 10^4$. The results in Figure~\ref{fig:efficiency} show that our method Panther++ is orders of magnitude faster than ReFex when the network grows larger, without the accuracy lost compared with ReFex.

From Table~\ref{tb:efficiency}, we can also see that RWR, TopSim and RoleSim cannot complete  top-$k$ similarity search for all vertices within a reasonable time when the number of edges increases to 500,000. ReFeX can deal with larger networks, but also fails when the edge number increases to 10,000,000. Our methods can scale up to handle very large networks with more than 10,000,000 edges. On average, \model only needs 0.0001 second to perform top-$k$ similarity search for each vertex in a large network. 
 
\hide{
\vpara{Scalability.}
In  the table of Figure~\ref{tb:efficiency}, we can see that when the network gets larger, some of comparative methods can not finish the task within several days, which are denoted as ``-'' in the table. Although TopSim save the time of pre-computational stage, it takes more than 120 hours for on-line query of all vertices when the number of edges is larger than 500,000. RoleSim calculates similarities between every pair of vertices in the pre-computational stage, which takes more than 72 hours when the number of edges is larger than 500,000. To improve the time efficiency of ReFex, we only extract the recursive features once, while the time cost still cannot be afforded when the number of edges is larger than 10,000,000. Our methods Panther and Panther++ only take XXX hours to return top-$k$ similar vertices for any vertex n the largest network with 300,000,000 vertices and 5,000,000,000 edges. The memory cost is linearly correlated with the number of edges, which is XXGB, 28.7GB, XXGB, 188GB, and 500GB in each network respectively.
}

\hide{
\begin{figure}
  \centering
  \graphicspath{ {./afbeeldingen/} }
  \graphicspath{ {./afbeeldingen/} }
  \includegraphics[width=0.35\textwidth]{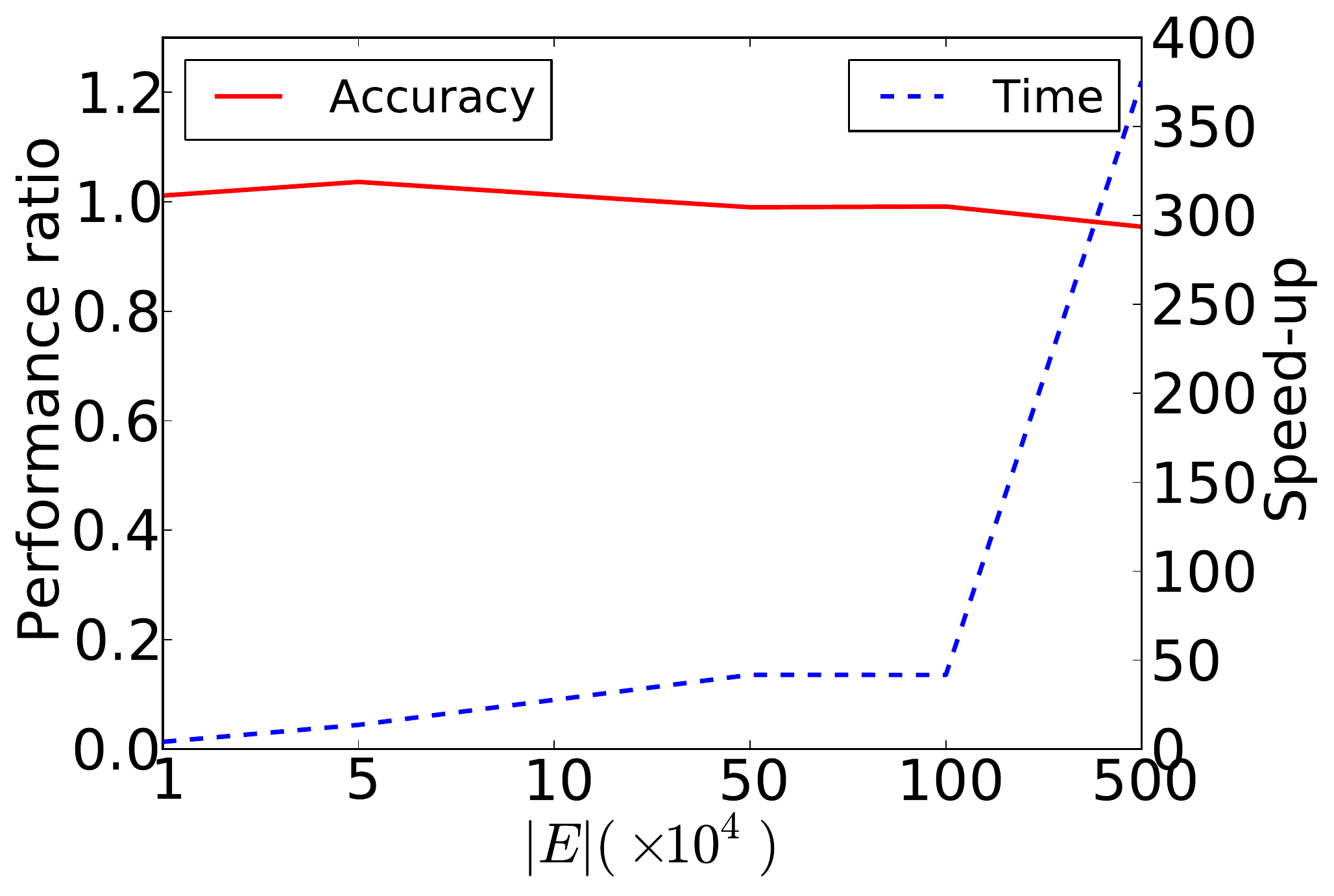}
  %\captionof{figure}{Fles 5}
\caption{\label{fig:speedupbak}  Performance ratio is calculated by $\frac{\text{Score(ReFex)}}{\text{Score(Panther++)}}$, where score is evaluated on the ground truth of structural holes (see Section~\ref{sec:accuracy} for details.). Speedup is calculated by $\frac{\text{Time(ReFex)}}{\text{Time(Panther++)}}$.}
\end{figure}
}

\begin{figure}
\centering
\mbox{
\hspace{-0.15in}
\subfigure[Speed-up]{\label{subfig:speeduprefex}
\includegraphics[width=0.245\textwidth]{Figures/timeaccuracy.pdf}
}
%\hspace{-0.1in}
\subfigure[Effect of $k$]{\label{subfig:parameterk}
\includegraphics[width=0.225\textwidth]{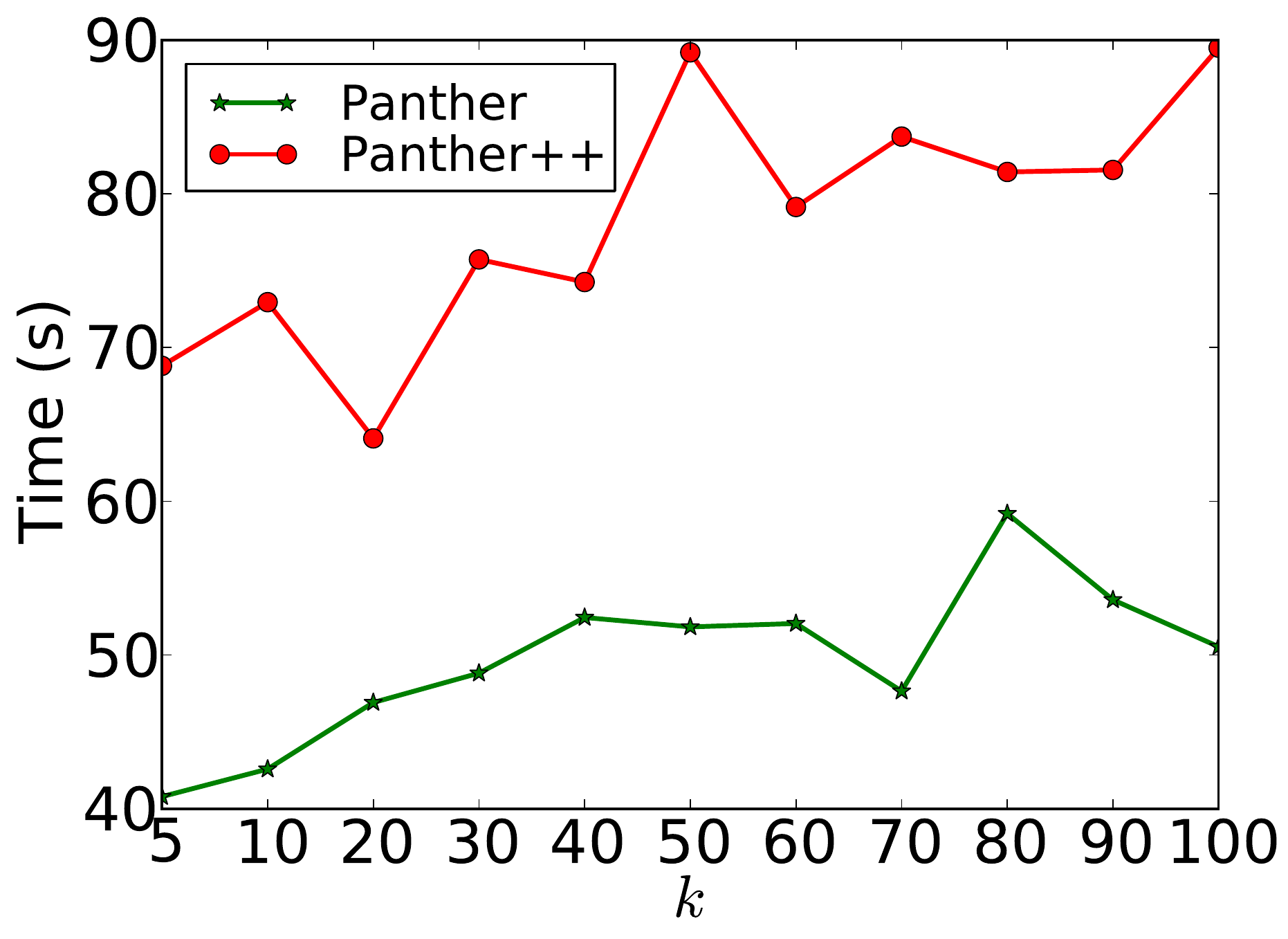}
}
}
%\vspace{-0.1in}
\caption{\label{fig:speedup} (a) Performance ratio is calculated by $\frac{\text{Score(ReFex)}}{\text{Score(Panther++)}}$, where score is evaluated by the application of structural hole spanner finding (see \S~\ref{sec:accuracy} for details.); Speed-up is calculated by $\frac{\text{Time(ReFex)}}{\text{Time(Panther++)}}$; (b) Effect of $k$ on the efficiency performance of Panther and Panther++.}
\end{figure}

\subsection{Accuracy Performance with Applications}
\label{sec:accuracy} 

\para{Identity Resolution.}
It is difficult to find a ground truth to evaluate the accuracy for similarity search. 
To quantitatively evaluate the accuracy of the proposed methods and compare with the other methods, we consider an application of \textit{identity resolution} on the co-author network. 
The idea is that we first use the authorship at different conferences to generate multiple co-author networks. An author may have  a corresponding vertex in each of the generated networks.
We assume that the same authors in different networks of the same domain are similar to each other. We anonymize author names in all the networks.
Thus given any two co-author networks, for example KDD-ICDM, we
perform a top-$k$  search to find similar vertices from ICDM for each vertex in KDD by different methods. 
If the returned $k$ similar vertices from ICDM by a method consists of the corresponding author of the query vertex from KDD, we say that the method hits a correct instance. 
A similar idea was also employed to evaluate similarity search in~\cite{gilpin:2013}.
Please note that the search is performed across two disconnected networks. 
%\para{Evaluation methods.}
%We evaluate Panther++ by resolving identity of authors from the KDD network to the authors in ICDM network~\cite{gilpin:2013}. For each author in both conferences, we select its vector calculated in KDD network and find the top-$k$ similar authors from ICDM network based on their vectors calculated in ICDM network. If the original author from KDD network is present in the set of $k$ most similar authors in ICDM network, we count the result as a match. In the same way, we resolve identity of authors from SIGIR to CIKM, and from SIGMOD to ICDE.  
%As the generated networks are disconnected, 
Thus, RWR, TopSim and RoleSim cannot be directly used for solving the task. ReFex calculates a vector for each vertex, and can be used here.
% in the same way to solve the problem. 
Additionally, we also compare with several other methods including Degree, Clustering Coefficient, Closeness, Betweenness and Pagerank. 
In our methods, Panther is not applicable to this situation. We only evaluate Panther++ here.
%For one given vertex in a network, each method matches the top-$k$ vertices  based on each of these metrics. In addition, 
Additionally, we also show the performance of random guess.
%compare the method Random, which selects the top-$k$ vertices randomly.
%
 %We also employ a similar way used in~\cite{jeh:KDD06} to quantitatively evaluate the performance of  Panther and Panther++.

%\vpara{Performance of Panther++.}
Figure~\ref{fig:nework_match} presents the performance of Panther++ on the task of identity resolution across co-author networks. 
%TopSim and RoleSim are omitted since they cannot be directly used to match vertices across two different networks. 
We see that Panther++ performs the best on all three datasets. ReFex performs comparably well; however, it is not very stable. In the SIGMOD-ICDE case, it performs the same as Panther++, while in the KDD-ICDM and SIGIR-CIKM cases, it performs worse than Panther++, when $k\le 60$.

\begin{figure*}
\centering
\subfigure[KDD-ICDM]{\label{subfig:network_match_kdd_icdm}
\includegraphics[width=0.275\textwidth]{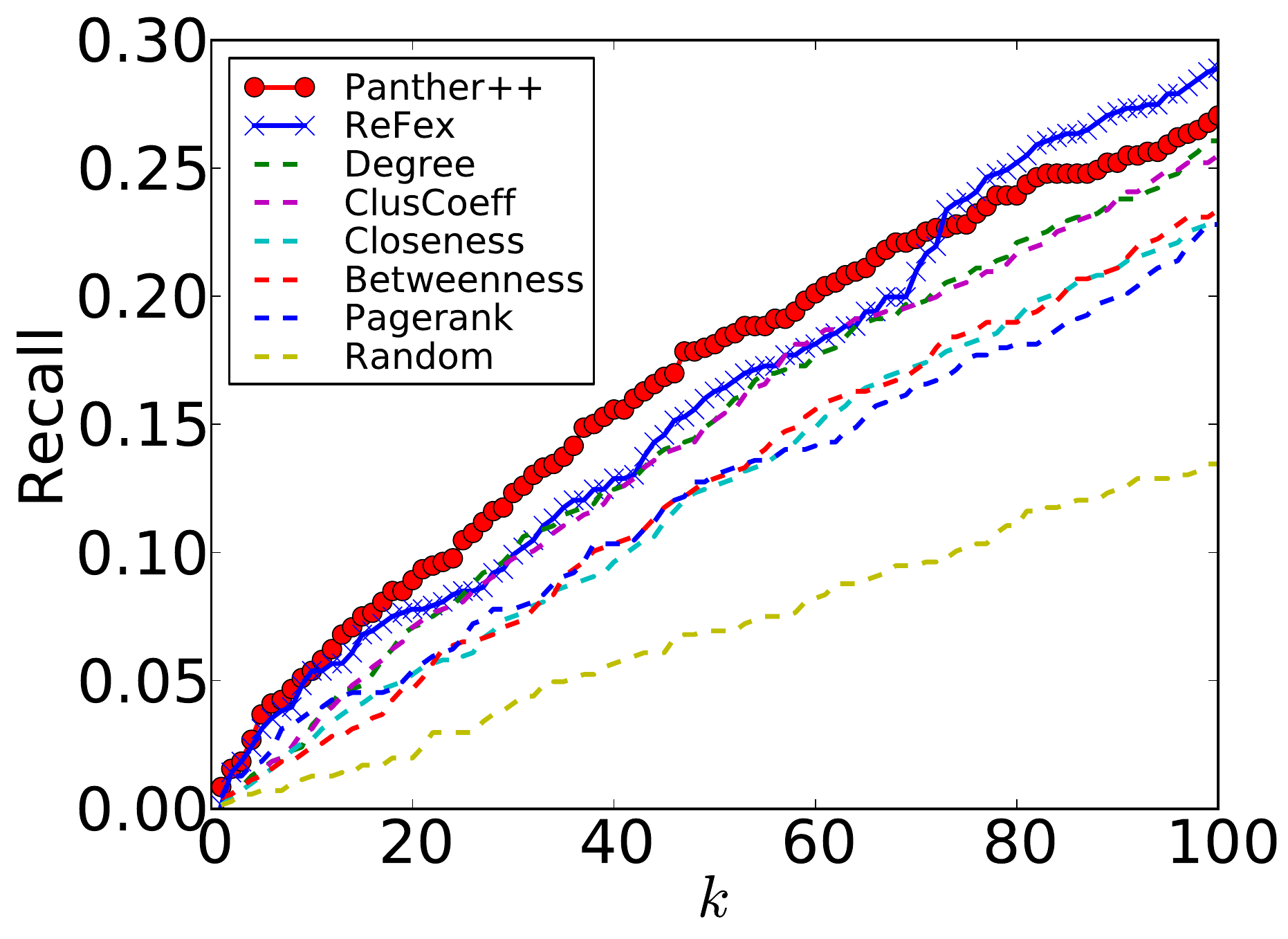}
}
\hspace{0.155in}
\subfigure[ SIGIR-CIKM]{\label{subfig:network_match_sigir_cikm}
\includegraphics[width=0.275\textwidth]{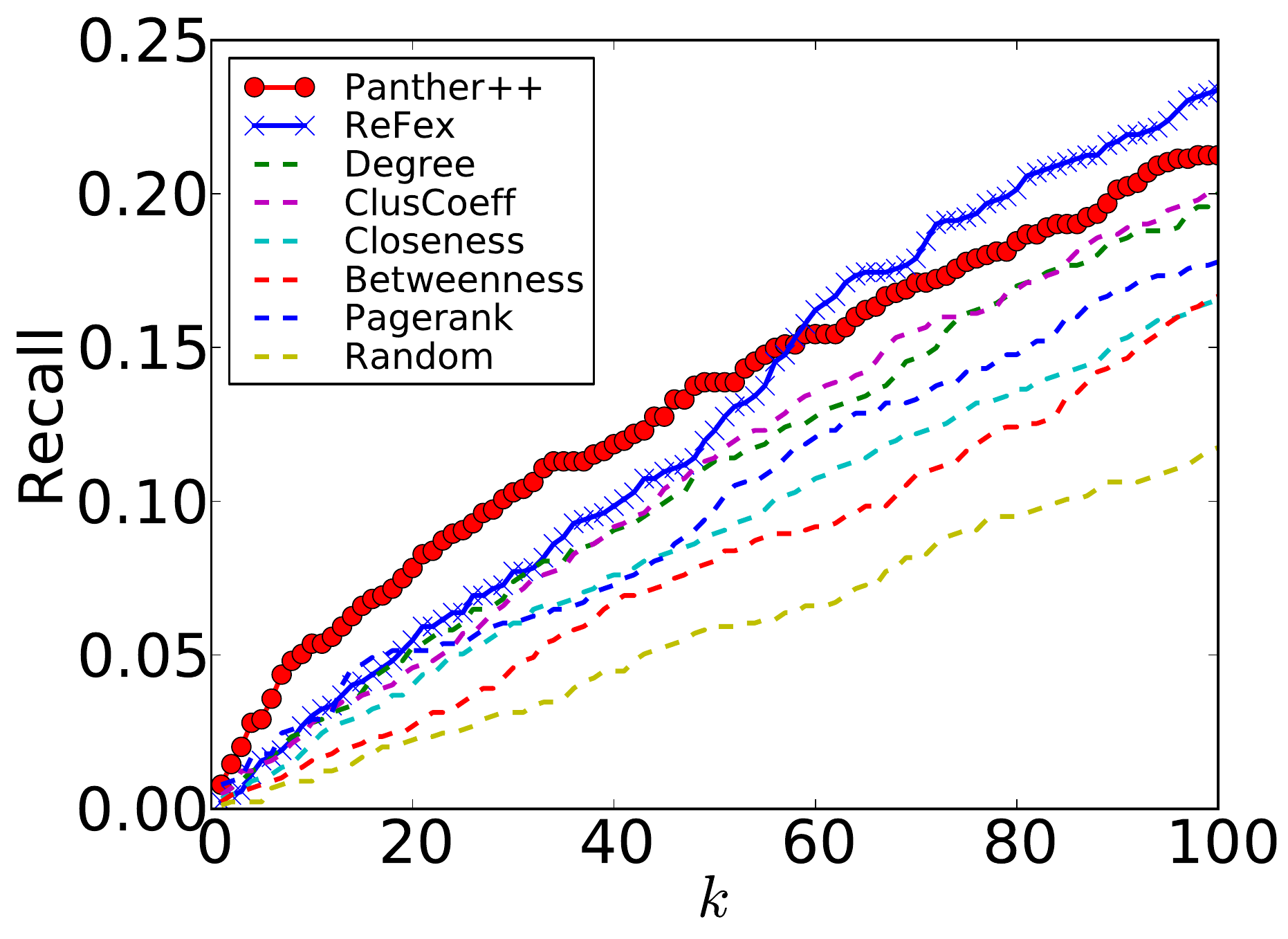}
}
\hspace{0.155in}
\subfigure[SIGMOD-ICDE]{\label{subfig:network_match_sigmod_icde}
\includegraphics[width=0.275\textwidth]{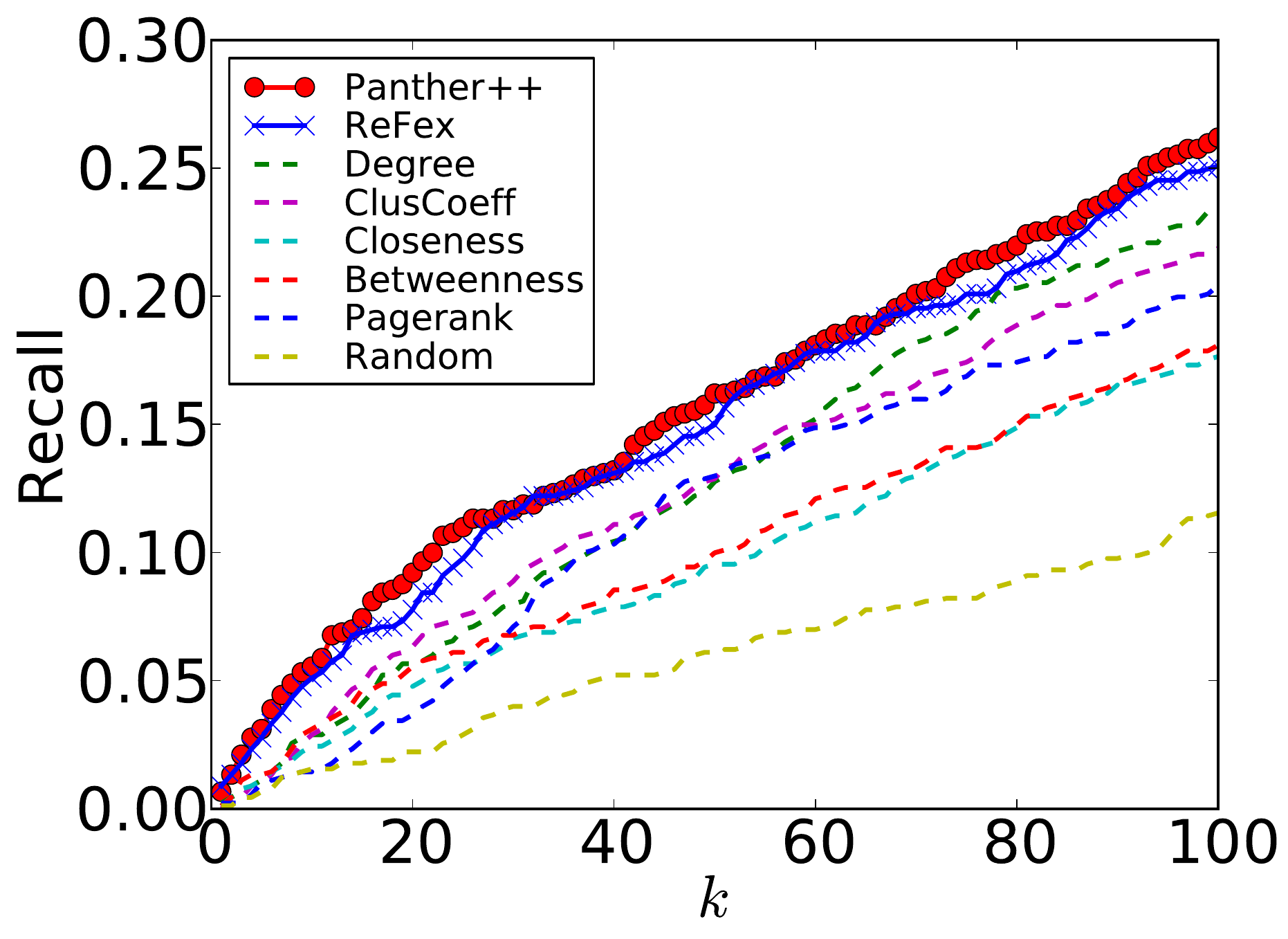}
}
%\vspace{-0.155in}
\caption{\label{fig:nework_match} Performance of identity resolution across two networks with different comparison methods.}
%Panther++ on the task of identity resolution across two networks.}
\end{figure*}

\vpara{Approximating Common Neighbors.}
We evaluate how Panther can approximate the similarity based on common neighbors.
%For evaluating the performance of Panther,  
The evaluation procedure is described as follows:
\begin{enumerate}
\item For each vertex $u$ in the seed set $S$, generate top $k$ vertices  $\text{Top}_{A,k}(u)$ that are the most similar to $u$ by the algorithm $A$.
\item For each vertex $v \in \text{Top}_{A, k}(u)$, calculate $g(u,v)$, where $g$ is a coarse similarity measure defined as the ground truth. Define $f_{A,k} = \sum_{u}\sum_{v}g(u,v)$.
\item Similarly, let $f_{R,k}$ denotes the result of a random algorithm. 
%That is $\text{Top}_{R,k}(u)$ contains $k$ random vertices selected for $u$.
\item Finally, we define the score for algorithm $A$ as $\text{score}(A,k) = \frac{f_{A,k} - f_{R,k}}{|S| \times k}$, which represents the improvement of algorithm $A$ over a random-based method. 
%Obviously, the higher the score, the more effective $A$ is.
\end{enumerate}

Specifically, we define $g(u,v)$ to be the number of common neighbors between $u$ and $v$ on each dataset.
%, as Panther is based on the principle that two vertices are considered structurally equivalent if they have many common neighbors.
%, Panther should at least rank higher the vertices with more common neighbors.

Figure~\ref{fig:pathsim} shows the performance of Panther evaluated on the ground truth of common neighbors in different networks. Some baselines such as RWR and RoleSim are ignored on different datasets, because they cannot complete  top-$k$ similarity search for all vertices within a reasonable time. It can be seen that Panther performs better than any other methods on most datasets. Panther++, ReFex and Rolesim perform worst since they are not devised to address the similarity between near vertices. Our method Panther performs as good as TopSim, the top-$k$ version of SimRank, because they both based on the principle that two vertices are considered structuraly equivalent if they have many common neighbors in a network. However, according to our previous analysis, TopSim performs much slower than Panther.

\begin{figure}
\centering
\mbox{
\hspace{-0.1in}
\subfigure[KDD]{\label{subfig:kdd_cn}
\includegraphics[width=0.23\textwidth]{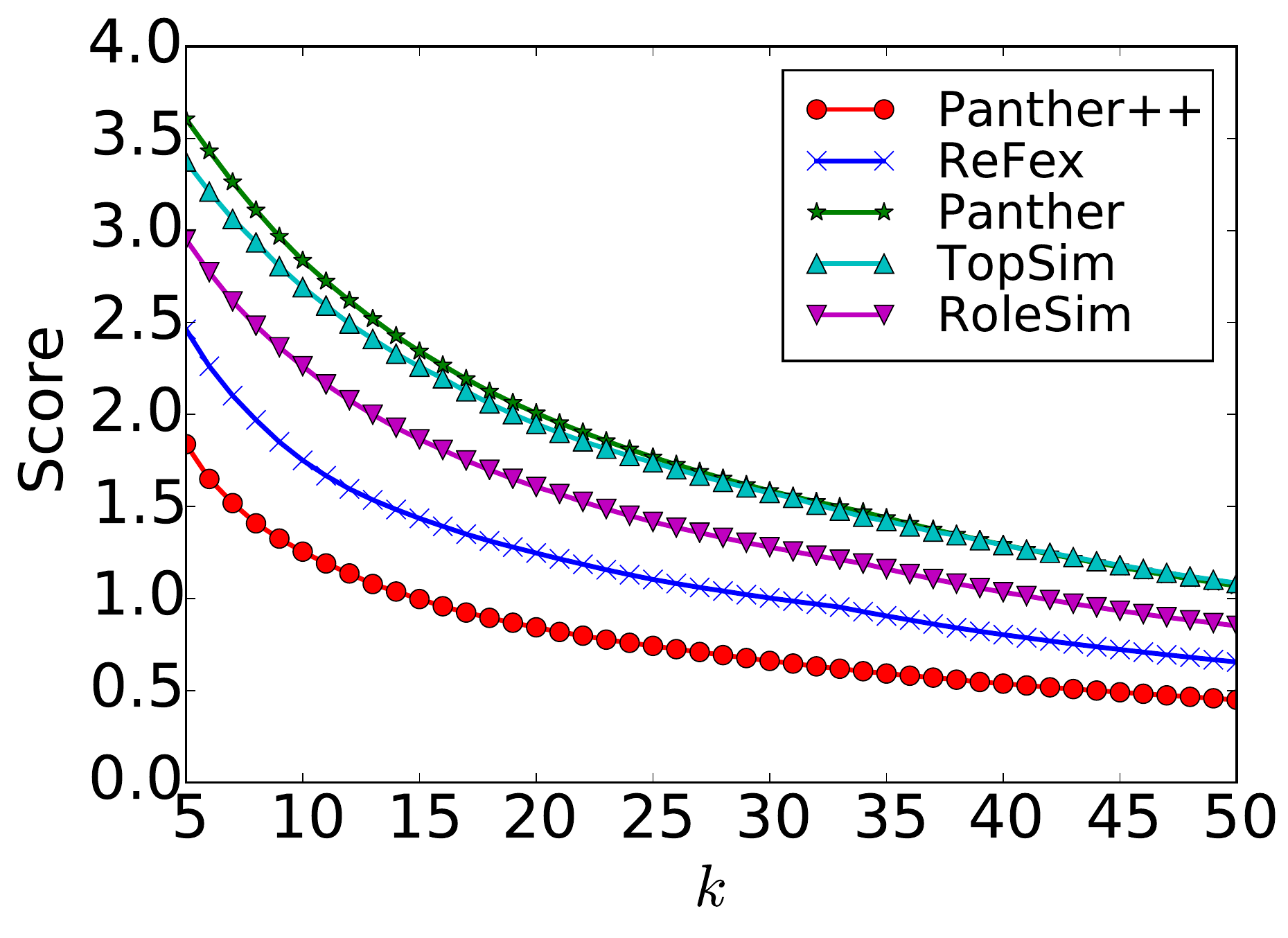}
}
%\hspace{-0.15in}
\subfigure[SIGIR]{\label{subfig:sigir_cn}
\includegraphics[width=0.23\textwidth]{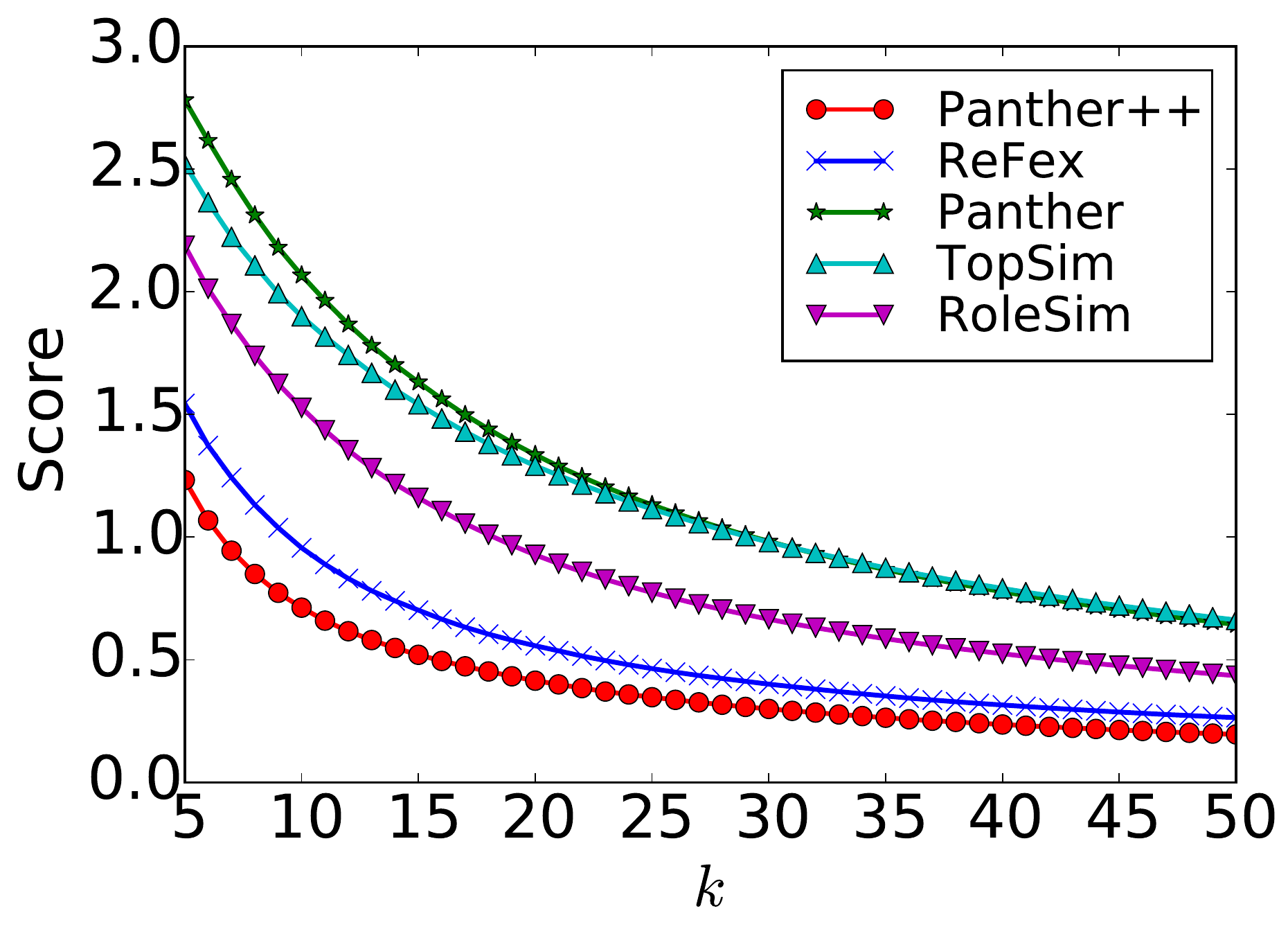}
}
}
\mbox{
\hspace{-0.1in}
\subfigure[Twitter]{\label{subfig:twitter}
\includegraphics[width=0.23\textwidth]{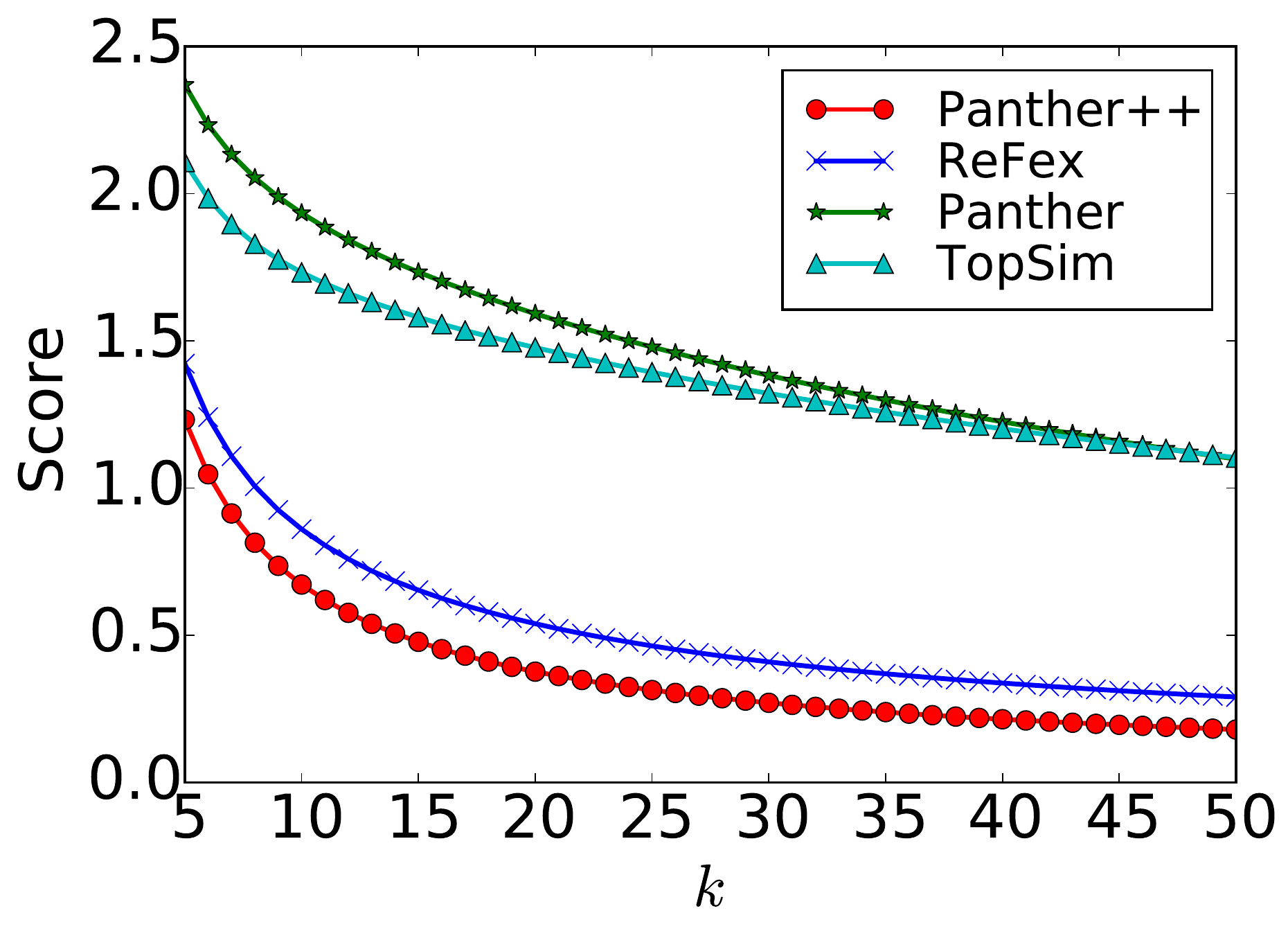}
}
%\hspace{-0.15in}
\subfigure[Mobile]{\label{subfig:mobile}
\includegraphics[width=0.23\textwidth]{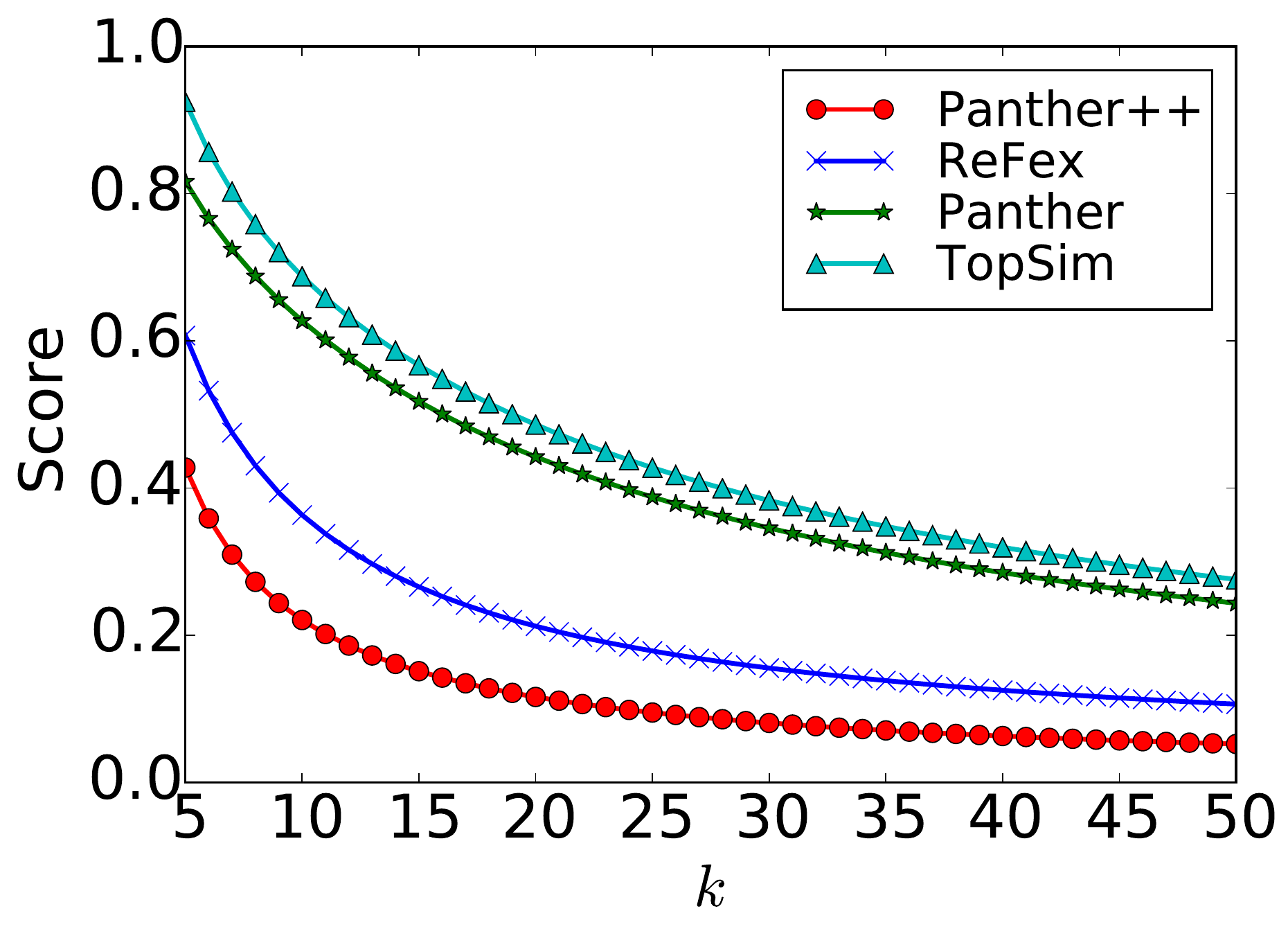}
}
}
\caption{\label{fig:pathsim} Performance of Panther evaluated on the ground truth of common neighbors.}
\end{figure}

%\vpara{Senior Program Member.}
\vpara{Top-$k$ Structural Hole Spanner Finding.}
The other application we consider in this work is top-$k$ structural hole spanner finding.
The theory of structural holes~\cite{burt:2009structural} suggests that, in social networks, individuals would benefit from filling the ``holes'' between people or groups that are otherwise disconnected. 
The problem of finding top-$k$ structural hole spanners was proposed in~\cite{LouTang:13WWW}, which also shows that 1\% of users who span structural holes control 25\% of the information diffusion (retweeting) in Twitter. 

Structural hole spanners are not necessarily connected, but they share the same structural patterns such as local clustering coefficient and centrality.
Thus, the idea here is to feed a few seed users to the proposed Panther++, and use  it to find other structural hole spanners.
For evaluation, we use network constraint~\cite{burt:2009structural} to obtain the structural hole spanners in Twitter and Mobile, and use this as the ground truth. Then we apply different methods---Panther++, ReFex, Panther, and SimRank---to retrieve top-$k$ similar users for each structural hole spanner. If an algorithm can find another structural hole spanner in the top-$k$ returned results, then it makes a correct search.  
We define $g(u,v) =1$, if both $u$ and $v$ are structural hole spanners, and $g(u,v) = 0$ otherwise.

\hide{
For evaluating the performance of Panther++, we use different ground truths for these datasets. In the co-author network, we define $g(u,v) = 1$ if both $u$ and $v$ are senior program committees of the same conference, and $g(u,v) = 0$ otherwise.
The underlying assumption is that the network structures of senior program committees are similar to each other. Earlier research~\cite{newman:2006} conducted in co-author networks observes that vertices with high degree or high betweenness would be likely to be high-impact authors. Usually, senior program committees are high-impact authors of a chosen field, thus to some degree share the same structural characteristics.
In the mobile, Twitter and Tencent networks, the ground truth is based on structural holes: $g(u,v) =1$ if both $u$ and $v$ are structural holes, and $g(u,v) = 0$ otherwise.
The structural holes are the kind of users acting as bridges in a network. We treat the users with top 20\% degree and bottom 20\% network constraint. Network constraint is widely used to measure the degree of structural holes~\cite{burt:2009structural} and is calculated as follows:

\beq{
    NC(v_i) = \frac{1}{\mathcal{N}(v_i)} (\sum_{ v_j\in \mathcal{N}(v_i)}  \sum_{v_k \in \mathcal{N}(v_j)} w_{j, k} \times w_{k,i})^2. \nonumber
}
}

Figure~\ref{fig:pathvec} shows the performance of 
%Panther++ evaluated on the ground truth of senior program committees in co-author networks and 
comparison methods for finding structural hole spanners in different networks. 
Panther++ achieves a consistently better performance than the comparison methods by varying the value of $k$.
TopSim, the top-$k$ version of SimRank seems inapplicable to this task. This is reasonable, as the underlying principle of SimRank is to find vertices with more connections to the query vertex.

%We can see from the results that Panther++ performs best on all the datasets. We do not plot the results of RoleSim in Twitter and Mobile networks because we cannot afford the computational costs of the algorithm in such large networks.
%These results indicate that Panther++ can effectively capture the structural characteristics, and thus reflect the network position or role of a vertex.

\begin{figure}
\centering
\mbox{
\hspace{-0.1in}

\subfigure[KDD]{\label{subfig:kdd_spc}
\includegraphics[width=0.23\textwidth]{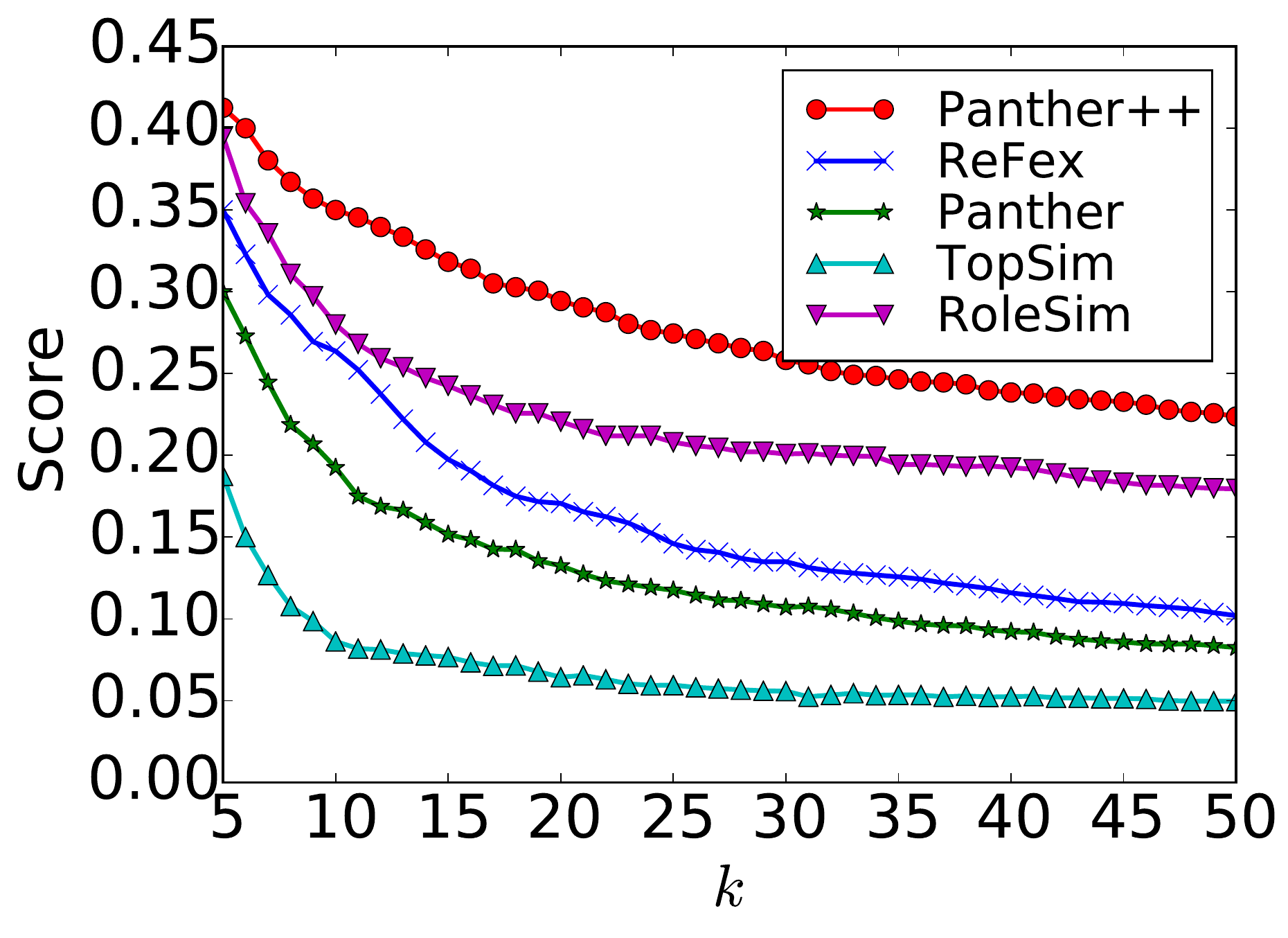}
}
%\hspace{-0.15in}
\subfigure[SIGIR]{\label{subfig:sigir_spc}
\includegraphics[width=0.23\textwidth]{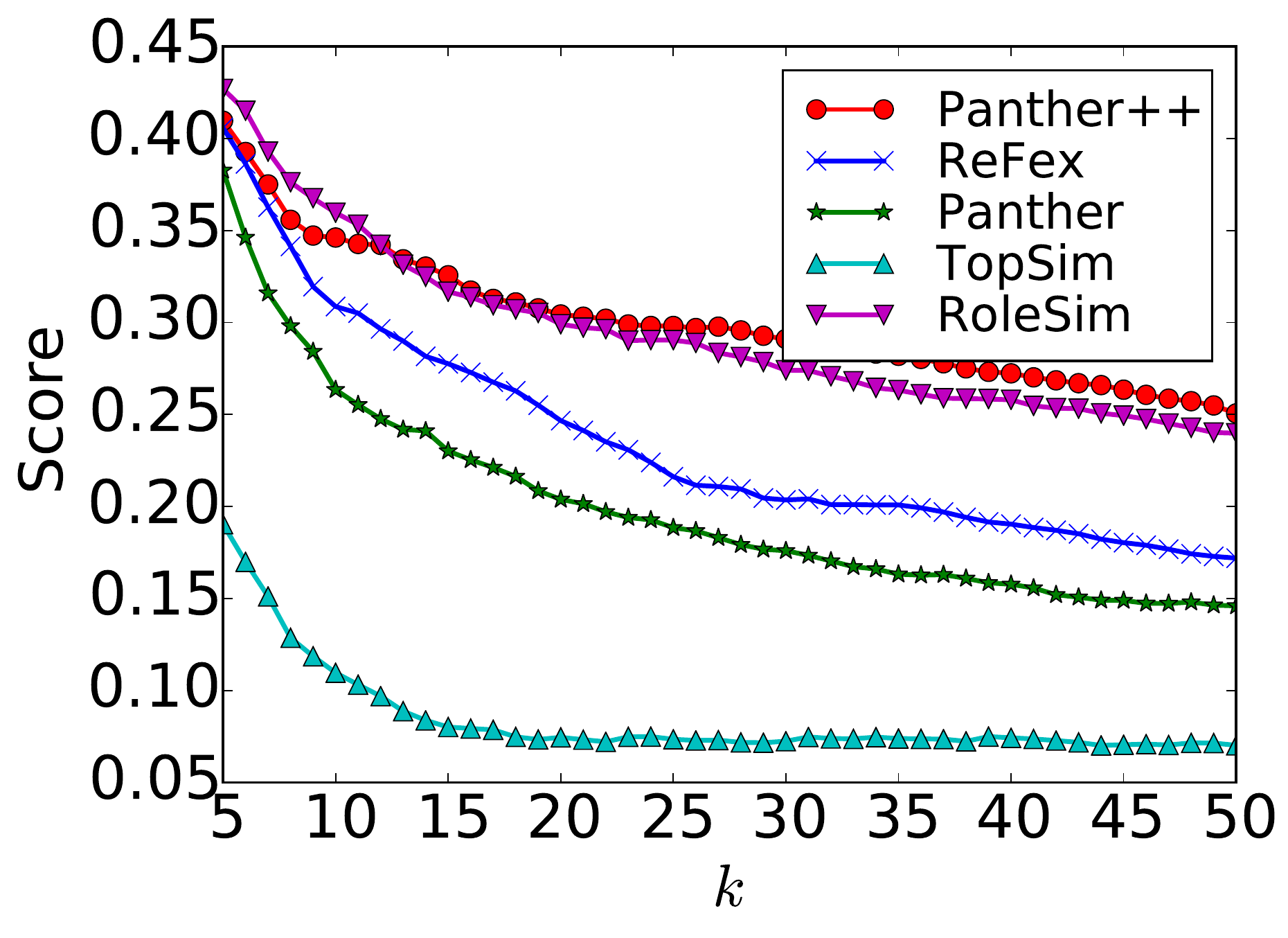}
}
}
\mbox{
\hspace{-0.1in}
\subfigure[Twitter]{\label{subfig:twitter_spc}
\includegraphics[width=0.23\textwidth]{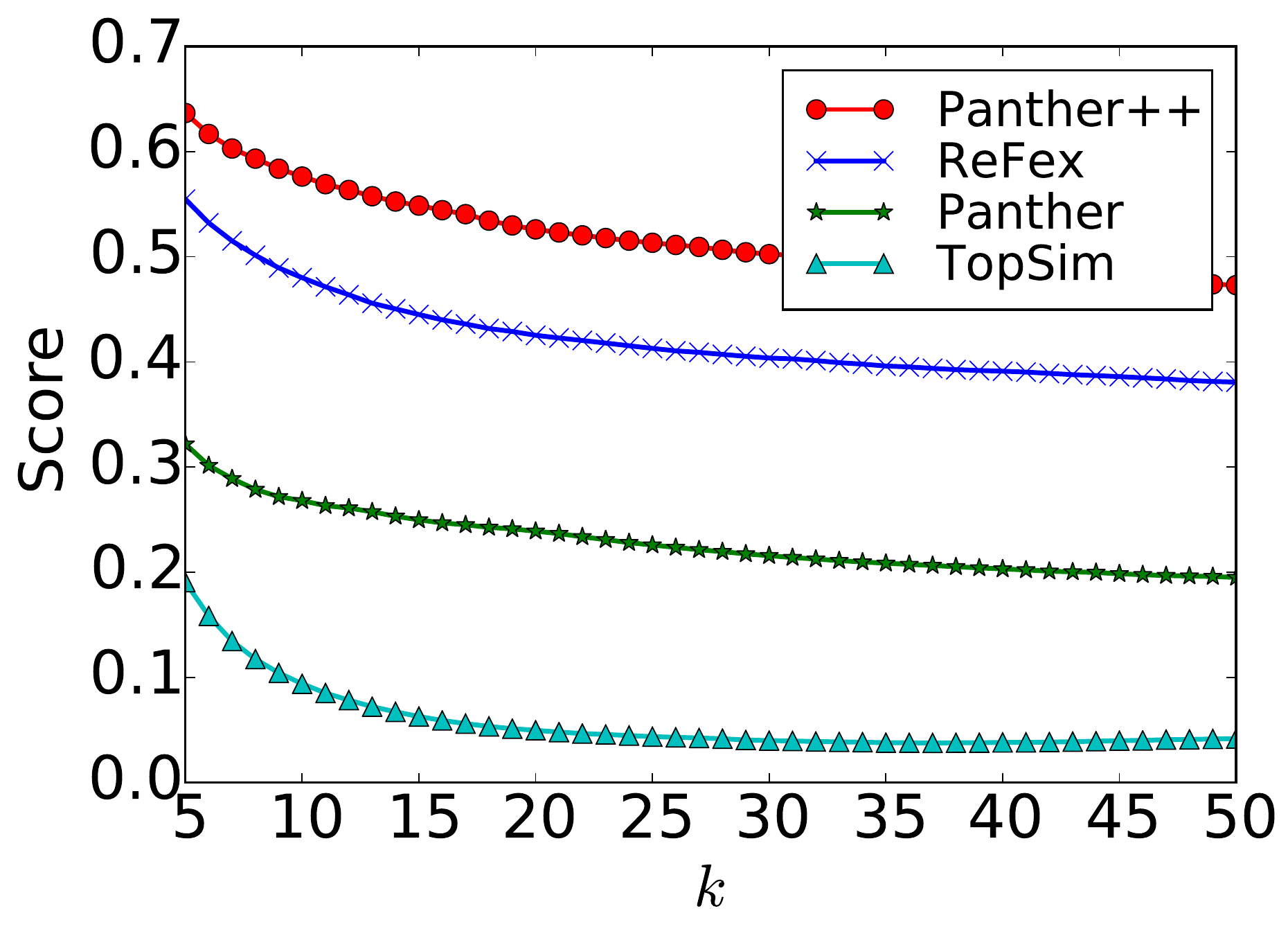}
}
%\hspace{-0.1in}
\subfigure[Mobile]{\label{subfig:mobile_spc}
\includegraphics[width=0.23\textwidth]{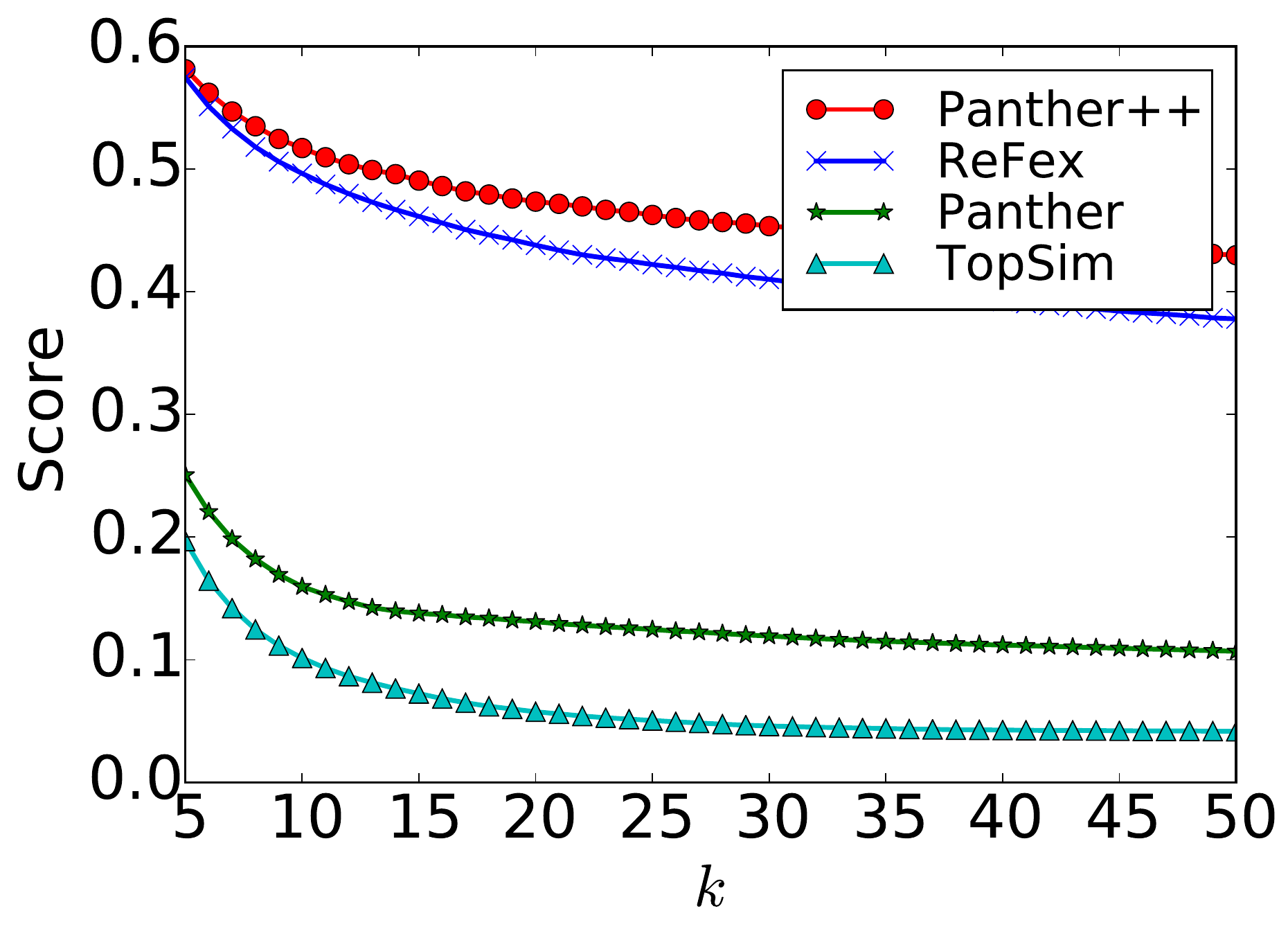}
}
}
%\vspace{-0.15in}
\caption{\label{fig:pathvec} Performance of mining structural hole spanners on the Twitter and Mobile networks with different methods.}
%Panther++ evaluated on the ground truth of senior program committees or structural holes.}
\end{figure}

\subsection{Parameter Sensitivity Analysis}

We now discuss how  different parameters influence the performance of our methods.

\vpara{Effect of Path Length $T$.}
%The path length $T$ directly affects the performance of Panther and Panther++.
%Figure~\ref{fig:T_pathsim} shows the accuracy performance of Panther on the ground truth of common neighbors by varying the values of path length $T$ as 2, 5, 10, 20, 50 and 100. 
%We can see that Panther gets better with the increase of path length $T$.  
Figure~\ref{fig:T_pathvec} shows the accuracy performance of Panther++ for mining structural holes by varying the path length $T$ as 2, 5, 10, 20 , 50 and 100.
A too small $T (< 5)$ would result in inferior performance. 
On Twitter, when increasing its value up to 5, it almost becomes stable. 
On Mobile, the situation is a bit complex, but in general $T=5$ seems to be a good choice.
% The result is rather varied and there is no universal law that can capture the influence of path $T$ on the performance. However, path length $T = 5$ performs consistently well on all the datasets.

\hide{
\begin{figure}
\centering
\subfigure[KDD]{\label{subfig:kdd_T_pathsim}
\includegraphics[width=0.23\textwidth]{Figures/kdd_T_neighborsim.pdf}
}
%\hspace{-0.15in}
\subfigure[SIGIR]{\label{subfig:sigir_T_pathsim}
\includegraphics[width=0.23\textwidth]{Figures/sigir_T_neighborsim.pdf}
}
\mbox{
\hspace{-0.10in}
\subfigure[Twitter]{\label{subfig:twitter_T_pathsim}
\includegraphics[width=0.23\textwidth]{Figures/twitter2_T_neighborsim.pdf}
}
%\hspace{-0.10in}
\subfigure[Mobile]{\label{subfig:mobile_T_pathsim}
\includegraphics[width=0.23\textwidth]{Figures/mobile2_T_neighborsim.pdf}
}
}
\hspace{-0.15in}
\caption{\label{fig:T_pathsim} Effect of path length $T$ on the performance of Panther.}
\end{figure}
}

\begin{figure}
\centering
%\subfigure[KDD]{\label{subfig:kdd_T_pathvec}
%\includegraphics[width=0.23\textwidth]{Figures/kdd_T_structuresim.pdf}
%}
%\hspace{-0.15in}
%\subfigure[SIGIR]{\label{subfig:sigir_T_pathvec}
%\includegraphics[width=0.23\textwidth]{Figures/sigir_T_structuresim.pdf}
%}
\mbox{
\hspace{-0.1in}
\subfigure[Twitter]{\label{subfig:sigmod_T_pathvec}
\includegraphics[width=0.23\textwidth]{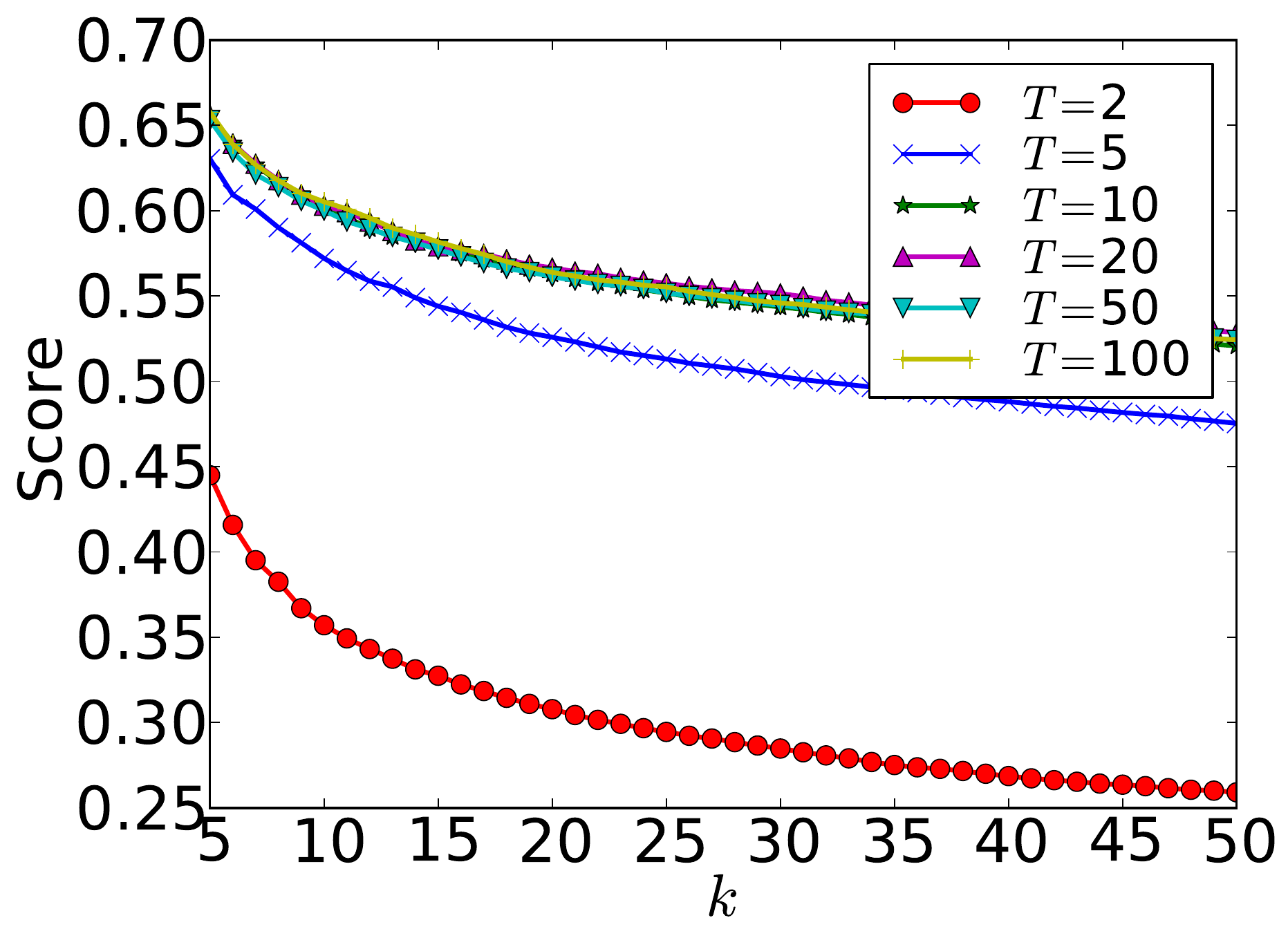}
}
%\hspace{-0.15in}
\subfigure[Mobile]{\label{subfig:twitter_T_pathvec}
\includegraphics[width=0.23\textwidth]{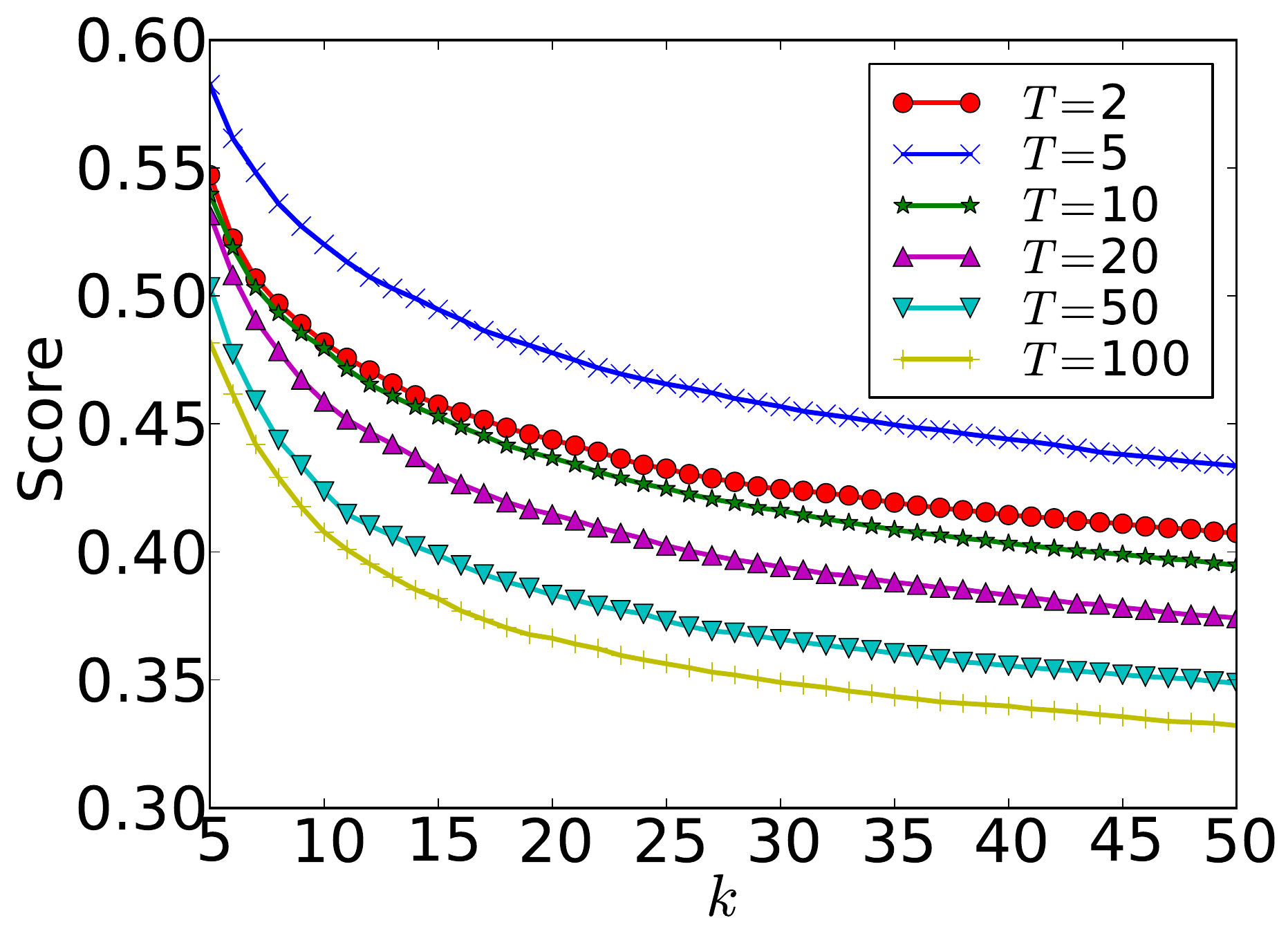}
}
}
%\vspace{-0.15in}
\caption{\label{fig:T_pathvec} Effect of path length $T$ on the accuracy performance of Panther++.}
\end{figure}

\vpara{Effect of Vector Dimension $D$.}
%The vector dimension $D$ only affects the result of Panther++.
Figure~\ref{fig:D_pathvec} shows the accuracy performance of Panther++ 
%on the ground truth of senior program committees in co-author networks and 
for mining structural hole spanners by varying vector dimension $D$ as 2, 5, 10, 20, 50 and 100. Generally speaking, the performance gets better when $D$ increases and it remains the same after $D$ gets larger than 50. This is reasonable,
%As we said earlier, we use the top-$k$ similarities calculated by
as Panther estimates the distribution of a vertex linking to the other vertices. Thus, the higher the vector dimension, the better the approximation. Once the dimension exceeds a threshold, 
 %nothing significant will be added, so 
 the performance gets stable. 
 %(Refer to Principal Component Analysis (PCA) for the same idea.)

\vpara{Effect of Error Bound $\varepsilon$.}
%The error-bound $\varepsilon$ affects the performance of both Panther and Panther++.
Figure~\ref{fig:Epsilon} shows the accuracy performance of Panther and Panther++ on the Tencent networks with different scales by varying error-bound $\varepsilon$ from 0.06 to 0.0001. We evaluate how Panther can estimate the similarity based on common neighbors.
%For evaluating the performance of Panther, 
Specifically, we use the same evaluation methods as structural hole spanner finding and define $g(u,v)$ to be the number of common neighbors between $u$ and $v$ on each dataset.
%We calculate the common neighbors of those vertices with other vertices and only select the structural holes within the sampled 10,000 vertices; Finally, we evaluate the top-$k$ results of PathSim and PathVec on the sampled 10,000 vertices. In this way, the scores of PathVec is not accurate, and thus decrease with the increase  of network scale. (???)
We  see that when the ratio $\frac{|E|}{(1/\varepsilon)^2}$ ranges from 5 to 20, scores of Panther are almost convergent on all the datasets. And when the ratio
$\frac{|E|}{(1/\varepsilon)^2}$ ranges from 0.2 to 5, the scores of Panther++ are almost convergent on all the datasets. Thus we can reach the conclusion that the value of $(1/\varepsilon)^2$ is almost linearly positively correlated with the number of edges in a network. Therefore we can empirically estimate $\varepsilon =   \sqrt{1/|E|}$ in our experiments.

\begin{figure}
\centering
%\subfigure[KDD]{\label{subfig:kdd_D_pathvec}
%\includegraphics[width=0.23\textwidth]{Figures/kdd_D_structuresim.pdf}
%}
%\hspace{-0.15in}
%\subfigure[SIGIR]{\label{subfig:sigir_D_pathvec}
%\includegraphics[width=0.23\textwidth]{Figures/sigir_D_structuresim.pdf}
%}
\mbox{
\hspace{-0.1in}
\subfigure[Twitter]{\label{subfig:twitter_D_pathvec}
\includegraphics[width=0.23\textwidth]{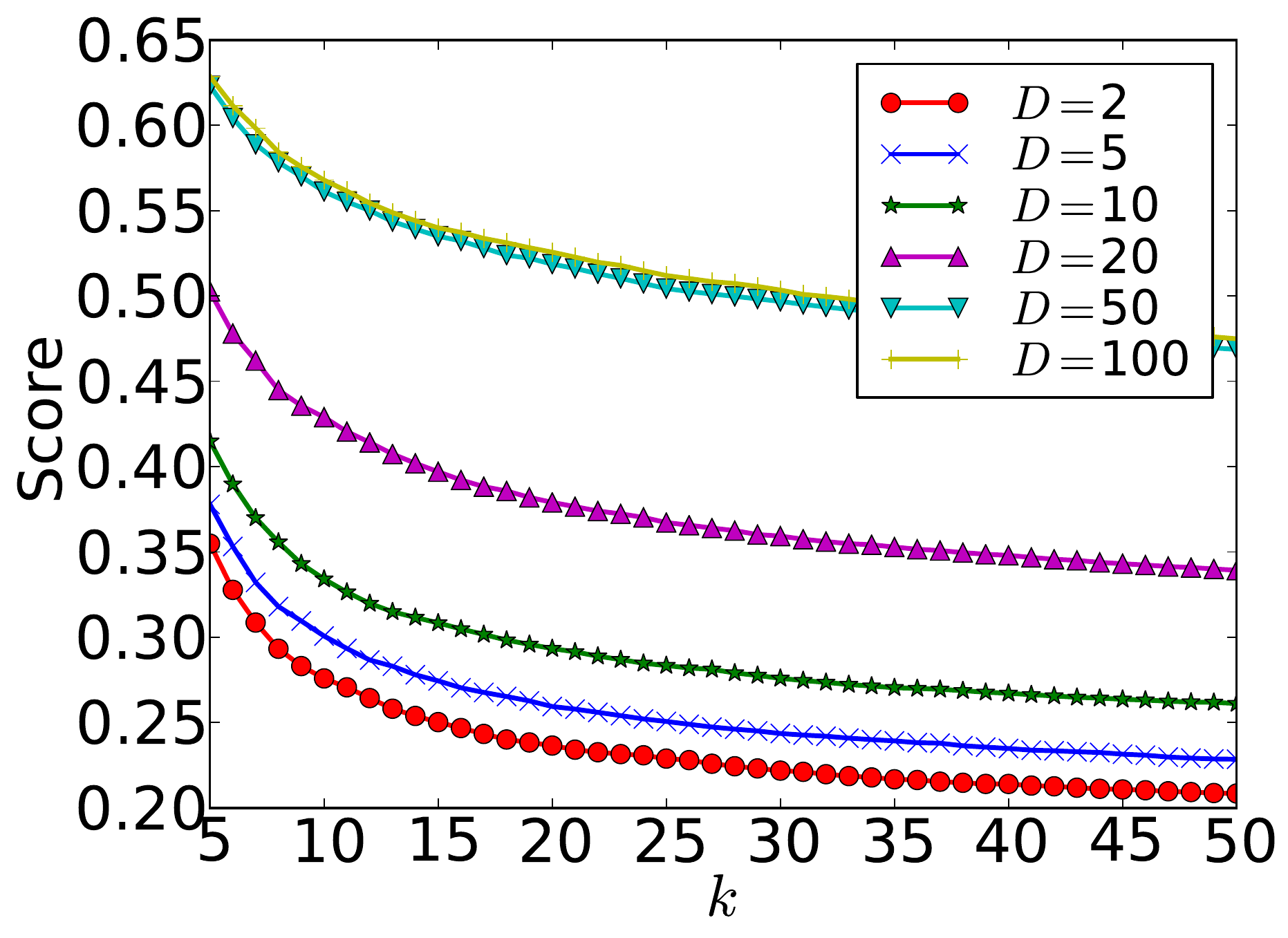}
}
%\hspace{-0.15in}
\subfigure[Mobile]{\label{subfig:mobile_D_pathvec}
\includegraphics[width=0.23\textwidth]{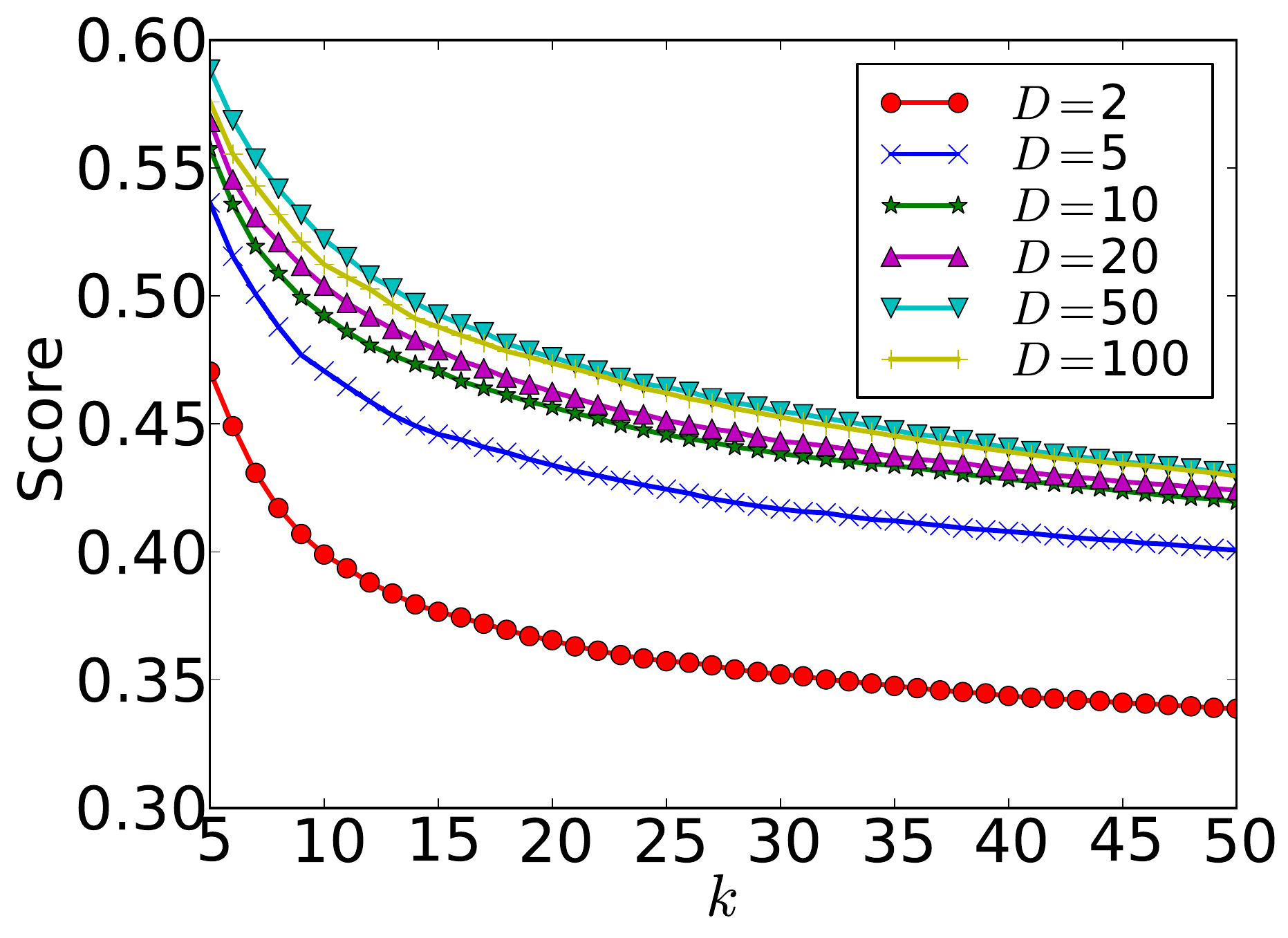}
}
}
%\vspace{-0.15in}
\caption{\label{fig:D_pathvec} Effect of vector dimension $D$ on the accuracy performance of Panther++.}
\end{figure}

\begin{figure}
\centering
\mbox{
\hspace{-0.1in}
\subfigure[Panther]{\label{subfig:epsilon_neighobrsim}
\includegraphics[width=0.23\textwidth]{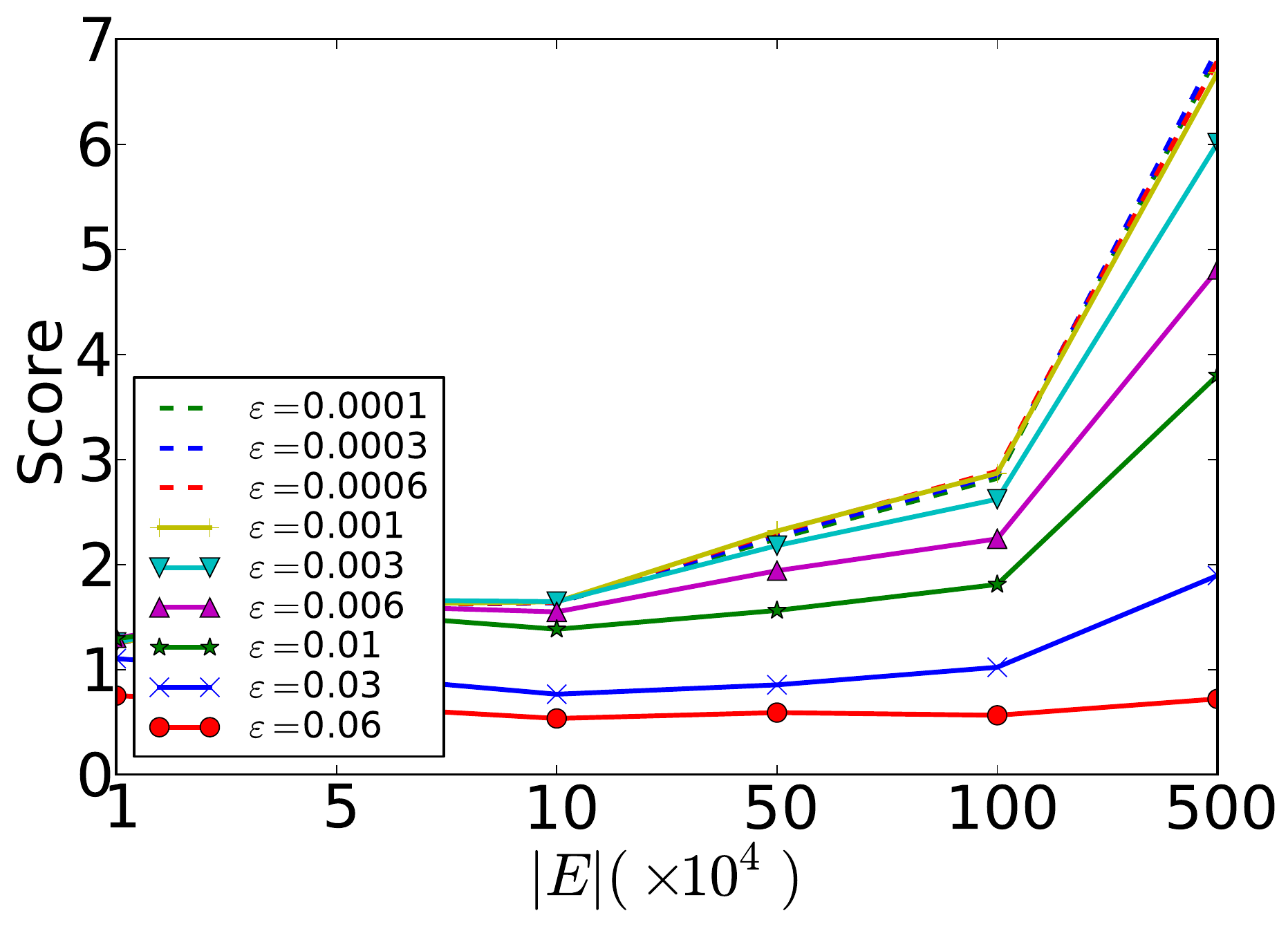}
}
%\hspace{-0.15in}
\subfigure[Panther++]{\label{subfig:epsilon_structuresim}
\includegraphics[width=0.23\textwidth]{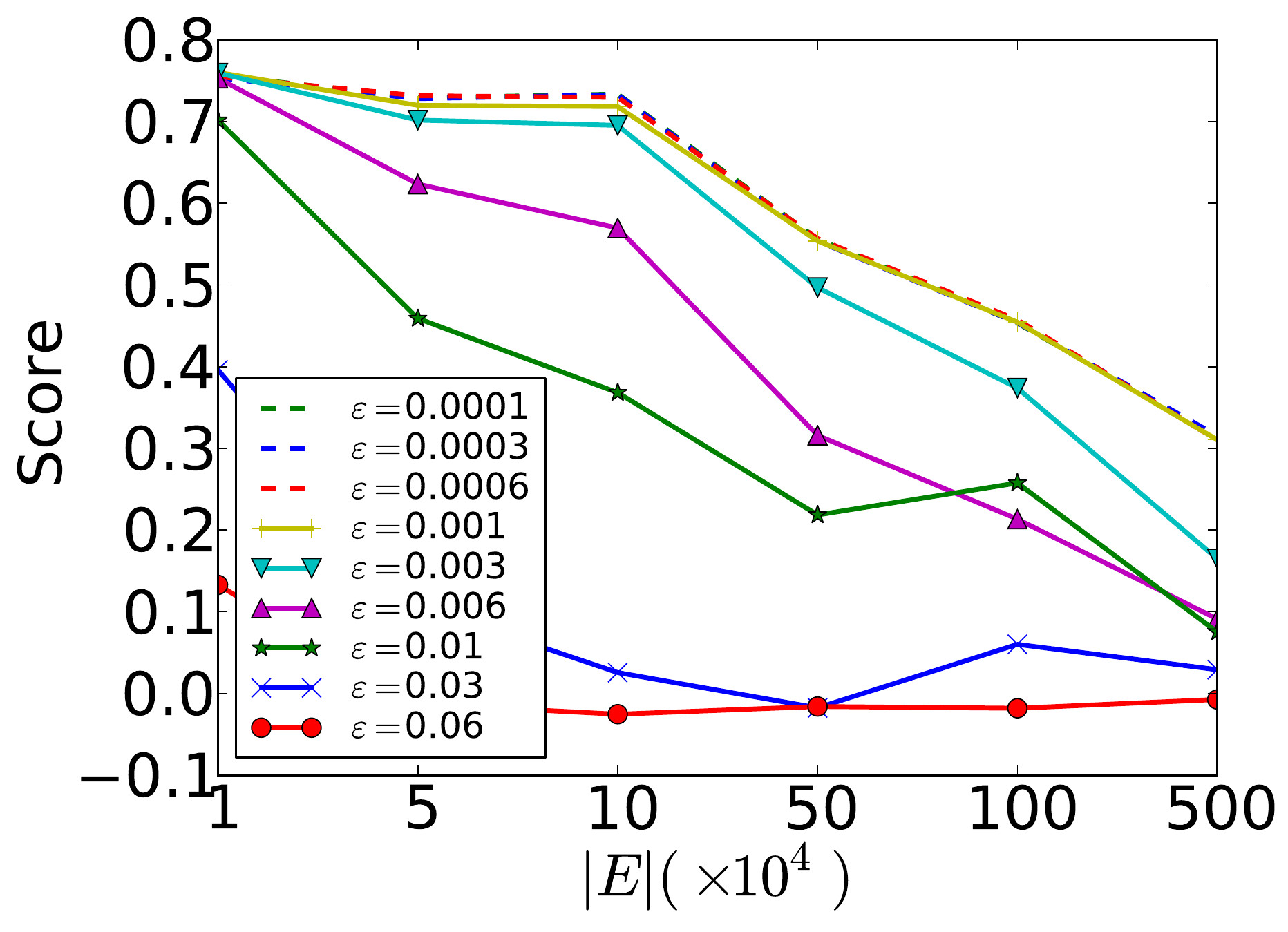}
}
}
%\vspace{-0.15in}
\caption{\label{fig:Epsilon} Effect of error-bound $\varepsilon$ on the performance of Panther and Panther++ on different sizes of Tencent networks.}
\end{figure}

\subsection{Qualitative Case study}

\begin{figure}[t]
\centering
\includegraphics[width=0.43\textwidth]{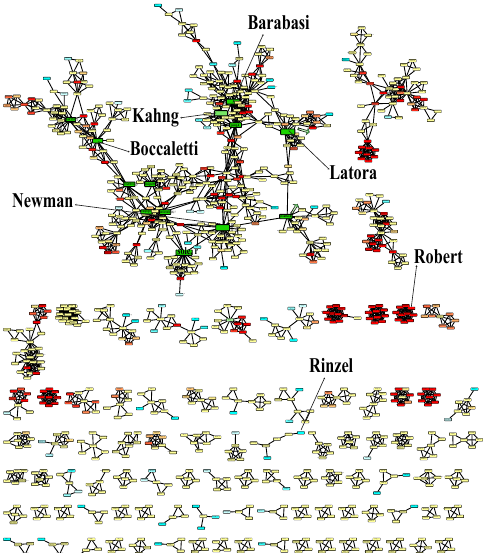}
%\vspace{-0.1in}
\caption{\label{fig:casestudy} Case study in a scientific co-author network~\cite{newman:2006}. The authors in similar positions to that of Barab\'{a}si are denoted in green,  similar to that of Robert are in red, and similar to that of Rinzel are in blue. Others are in yellow.}
\end{figure}

%\vpara{co-author network.}
%We employ Panther and Panther++ in the co-author networks, specifically KDD, ICML and SIGIR networks. 
Now we present two case studies to demonstrate the effectiveness of the proposed methods.

\vpara{``Similar Researchers''}
Table~\ref{tb:casestudy} shows an example of top-5 similar authors  to Jiawei Han, Michael I. Jordan, and W. Bruce Croft, found by Panther and Panther++. 
%As can be seen from the results, the 
The two methods present very different results.
Those authors found by Panther have closer connections with the query author. 
%As expected, the people 
While those authors found by Panther++ 
have a similar ``social status'' (essentially similar structural patterns) to the query author.
%are those who have high impacts in the domain thus share the same structures. 
For example, Philip S. Yu and Christos Faloutsos are two researchers as famous as Jiawei Han in the data mining field (KDD). Andrew Y. Ng and  Bernhard Scholkopft are  influential researchers similar to Michael I. Jordan in the machine learning field (ICML). 
%Leif Azzopardi and Maarten de Rijke are both experts similar as W.Bruce Croft in SIGIR domain.  

\begin{table*}

\centering
\caption{\label{tb:casestudy}  Case study of top-5 similar authors in KDD, ICML and SIGIR networks.}
\begin{tabular}{ c|c|c|c|c|c}
\hline
\multicolumn{2}{c|}{  \textbf{Jiawei Han} } & \multicolumn{2}{c|}{  \textbf{Michael I. Jordan} }& \multicolumn{2}{c}{\textbf{W. Bruce Croft} }  \\ \hline

Panther      & Panther++              & Panther                  & Panther++            &  Panther               & Panther++      \\ \hline
Chi Wang       & Philip S. Yu                & Eric p. Xing               & Andrew Y. Ng            &  Michael Bendersky  & Leif Azzopardi  \\
Jing Gao        &  Christos Faloutsos    & Percy Liang               & Bernhard Scholkopf   & Trevor Strohman      & Maarten de Rijke  \\
Xifeng Yan    & Jeping Ye                    & Lester W. Mackey      & Zoubin Ghahramani   & Jangwon Seo            & Zheng Chen  \\
YiZhou Sun   & Naren Ramakrishnan  & Gert R. G. Lanckriet   & Michael l. Littman     & Donald Metzler        & Ryen w. White  \\
Philip S. Yu    & Ravi Kumar                & Purnamrita Sarkar      & Thomas G. Dietterich& Jiwoon Jeon             & Chengxiang Zhai  \\

  \hline
\end{tabular}
\end{table*}

\vpara{``Who is similar to Barab\'{a}si?''}
Albert-L\'{a}szl\'{o} Barab\'{a}si is a famous Hungarian-American physicist, who proposed the Barab\'{a}si--Albert (BA) model for generating random scale-free networks using a preferential attachment mechanism.
%\vpara{Scientific network.}
We apply Panther++ to a scientific network~\cite{henderson:KDD12, newman:2006} to find 
 researchers who have similar structural positions to that of Dr. Barab\'{a}si.
It is interesting that different researchers play different roles in the network. 
Mark Newman and Vito Latora have similar structural patterns to that of Dr. Barab\'{a}si. Some other researchers like Robert form a tight-knit group with him.
%authors with similar structural position for three target authors: Newman, Robert, and Rinzel. The primary roles of the three authors are different. Newman acts as a broker; Robert places him in a tight-knit group, and Rinzel is in the periphery place. In Figure~\ref{fig:casestudy}, we present the names of the three target authors. We also denote the names of several authors in the center places. The authors similar as Newman are denoted in green, similar as Robert are in red. and similar as Rinzel are in blue. Other authors are in yellow. It can be clearly seen from the case study that 
Panther++ successfully recognizes those researchers with similar structural positions. 
%The similar results returned by ReFex are shown in \cite{henderson:KDD12}.

%% file: related.tex
\section{Related Work}
\label{sec:related}

Early similarity measures, including bibliographical coupling~\cite{Kessler:63} and co-citation~\cite{Small:73} are based on the assumption that two vertices are similar if they have many common neighbors. 
This category of methods cannot estimate similarity between vertices without common neighbors.
%A direct consequence is that two nodes with no common neighbors will be treated as not similar at all. Thus, s
Several measures have been proposed to address this problem. For example, Katz~\cite{katz:1953} counts two vertices as similar if there are more and shorter paths between them. Tsourakakis et al.~\cite{tsourakakis2014toward} learn a low-dimension vector for each vertex  from the adjacent matrix and calculate similarities between the vectors. 
Jeh and Widom~\cite{jeh:KDD06} propose a new algorithm, SimRank.
%, which not only considers common neighbors but also look for neighbors that are similar. They 
The algorithm follows a basic recursive intuition that two nodes are similar if they are referenced by similar nodes.
VertexSim~\cite{newman:2006} is an extension of SimRank.
% are refined versions of SimRank. They deal with the drawback of SimRank that when two vertices have same neighbors in common, their SimRank score will go down with the increase of the number of neighbors.
%Leicht et al.~\cite{newman:2006} develop an asymmetrical version of SimRank named vertex similarity. It is based on the assumption that two vertices are similar if any pair of their neighbors are similar. 
However, all the SimRank-based methods share a common drawback: their computational complexities are too high. 
%For example, SimRank requires $O(IN^2d^2)$ time and $O(N^2)$ space, where $I$ is the number of iterations, $N$ is the number of vertices and $d$ is the average degree over all vertices. 
Further studies have been done to reduce the computational complexity of SimRank~\cite{He:2010,kusumoto:2014scalable,lee:2012top}.  Fast-random-walk-based graph similarity, such as in~\cite{Fujiwara:2013, Sarkar:2010}, has also been studied recently. 
Sun et al.~\cite{sun2011pathsim} measure similarities between vertices based on their inter-paths instantiated from different schemes defined in a heterogeneous information network. The setting is different from ours and the algorithm is not efficient.

%Although SimRank is argued to measure structure similarity, the two vertices can not far away from each other. 

Most aforementioned methods cannot handle similarity estimation across different networks.
Blondel et al.~\cite{blondel2004measure} provide a HITS-based recursive method to measure similarity between vertices across two different graphs. RoleSim~\cite{jin:KDD11} can also calculate the similarity between disconnected vertices.  Similar to SimRank, the computational complexity of the two methods is very high.
Feature-based methods can match vertices with similar structures. 
%The basic idea is to define several features for each vertex and then calculate the Euclidean distance between feature vectors of two vertices as their structural similarity. 
%The survey about feature based methods can be refer to~\cite{rossi2014role}.
For example,
Burt~\cite{burt1990detecting} counts the 36 kinds of triangles in one's ego network to represent a vertex's structural characteristic. In the same way, vertex centrality, closeness centrality, and betweenness centrality~\cite{freeman1977set} of two different vertices can be compared, to produce a structural similarity measure.
% However, any of the above metrics is too limited to explain the structure characteristic of a vertex.
Aoyama et al.~\cite{aoyama2011fast} present a fast method to estimate similarity search between objects, instead of vertices in networks.
ReFex~\cite{henderson:KDD11,henderson:KDD12} defines basic features such as degree, the number of within/out-egonet edges,  and define the aggregated values of these features over neighbors as recursive features. The computational complexity of ReFex depends on the recursive times. 
More references about feature-based similarity search in networks can be found in the  survey~\cite{rossi2014role}.
%Although they use pruning to reduce the complexity, no theoretical proof is given to show how many recursive times is enough. 

%% file: conclusion.tex
\section{Conclusion}
\label{sec:conclusion}
\hide{
In this paper, we propose two methods PathSim and PathVec to correspondingly tackle the problems of neighborhood similarity and structure similarity. $PathSim$ is based on the simple assumption that ``two nodes are similar if they frequently appear on the same path'' and PathVec relies on the idea that ``vertices with similar topological characteristics will have similar distributions of tie strengths with other vertices''. In order to deal with the intractability of the calculation, we adopt a sampling based method to approximate the result. Theoretical justifications have been given for the use this approximation algorithm. Moreover, it shows that the number of random paths needed to achieve a desired approximation level is independent of the network size. The experiments on several different datasets demonstrate the effectiveness and efficiency of our proposed algorithms.
}

In this paper, we propose a sampling method to quickly estimate top-$k$ similarity search in large networks. 
% and accurately approximates the similarity between vertices. 
The algorithm is based on the idea of \textit{random path} and an extended method is also presented to enhance the structural similarity when two vertices are completely disconnected.
We provide theoretical proofs for the error-bound and confidence of the proposed algorithm.
We perform an extensive empirical study and show that our algorithm can obtain top-$k$ similar vertices for any vertex in a network approximately 300$\times$ faster than state-of-the-art methods.
We also use identity resolution and structural hole spanner finding, two important applications in social networks, to evaluate the accuracy of the estimated similarities. Our experimental results demonstrate that the proposed algorithm achieves clearly better performance than several alternative methods.

%% file: ack.tex
%\small

%\section*{ACKNOWLEDGMENTS}
\vpara{Acknowledgements.}
The authors thank Pei Lee, Laks V.S. Lakshmanan, Jeffrey Xu Yu; Ruoming Jin,  Victor E. Lee, Hui Xiong; Keith Henderson, Brian Gallagher, Lei Li, Leman Akoglu, Tina Eliassi-Rad, Christos Faloutsos for sharing codes of the comparation methods. We thank Tina Eliassi-Rad for sharing datasets.
\hide{
The work is supported by the
%jie's new 863 on social network
National High-tech R\&D Program (No. 2014AA015103),
%Maosong's  973
National Basic Research Program of China (No. 2014CB340500,
%Xiaofeng's 973
No. 2012CB316006),
%National High Technology Research and Development Program of China (No. 2014AA015103)
%jie's excellent young 
Natural Science Foundation of China (No. 61222212),
%jie's new one on social network mining
%No. 61073073, 
%2007 on semantic annotation
%, No. 60703059
%juanzi's
%No. 60973102,
%chunxiao's
%No. 61170061),

%shiqiang
%Chinese National Key Foundation Research (No. 60933013,
%juanzi's key
%No.61035004),
%a special fund for Fast Sharing of Science Paper in Net Era by CSTD, 
%chunxiao's 973
%National Basic Research Program of China (No. 2011CB302302),
%and a research fund of Tsinghua-Tencent Joint Laboratory for Internet Innovation Technology.
%slone
the Tsinghua University Initiative Scientific Research Program (20121088096),
a research fund supported by Huawei Inc.
and 
Beijing key lab of networked multimedia.

}
%自主科研? lab

%jie's 863
%National High-tech R\&D Program (No. 2009AA01Z138).
%tencent

%jie's
%and Chinese Young Faculty Research Fund (No. 20070003093).

%973
%No. 2007CB310803, 

%He Gao Ji
%2011ZX01042-001-002

\normalsize